\documentclass[a4paper,onecolumn,11pt,unpublished]{quantumarticle}
\pdfoutput=1

\usepackage[utf8]{inputenc}
\usepackage[english]{babel}
\usepackage[T1]{fontenc}
\usepackage{amsmath,amssymb,amsthm,mathtools}
\usepackage{bm}
\usepackage{enumitem}
\usepackage{microtype}
\usepackage[numbers]{natbib}
\usepackage[colorlinks=true,linkcolor=blue,citecolor=blue,urlcolor=blue]{hyperref}
\usepackage{orcidlink}
\usepackage{subcaption}

\allowdisplaybreaks

\newtheorem{theorem}{Theorem}[section]
\newtheorem{lemma}[theorem]{Lemma}
\newtheorem{proposition}[theorem]{Proposition}
\newtheorem{corollary}[theorem]{Corollary}
\newtheorem{fact}[theorem]{Fact}
\theoremstyle{definition}

\theoremstyle{remark}
\newtheorem{remark}[theorem]{Remark}

\newcommand{\C}{\mathbb C}
\newcommand{\R}{\mathbb R}
\newcommand{\Z}{\mathbb Z}
\newcommand{\F}{\mathbb F}
\newcommand{\cH}{\mathcal H}
\newcommand{\cK}{\mathcal K}
\newcommand{\cI}{\mathcal I}
\newcommand{\cS}{\mathcal S}
\newcommand{\cW}{\mathcal W}
\newcommand{\Cl}{\mathrm{Cliff}}

\newcommand{\Sp}{\operatorname{Sp}}
\newcommand{\Tr}{\operatorname{Tr}}

\newcommand{\vol}{\operatorname{vol}}
\newcommand{\supp}{\operatorname{supp}}
\newcommand{\Span}{\operatorname{span}}
\newcommand{\Stab}{\operatorname{Stab}}
\newcommand{\im}{\operatorname{im}}

\newcommand{\dist}{\operatorname{dist}}
\newcommand{\ket}[1]{\lvert #1\rangle}
\newcommand{\bra}[1]{\langle #1\rvert}

\newcommand{\proj}[1]{\lvert #1\rangle\!\langle #1\rvert}
\newcommand{\eps}{\varepsilon}
\newcommand{\Ncir}{N_{\mathrm{cir}}}
\newcommand{\BE}{\mathrm{BE}}

\begin{document}
\title{Optimal T-Count for Block Encodings of
Fermionic and Spin Hamiltonians}

\author{Jiaxin Ma}
\affiliation{Department of Electrical and Computer Engineering, North Carolina State University, Raleigh, NC 27606, USA}

\author{Kevin J. Joven}
\affiliation{Department of Electrical and Computer Engineering, North Carolina State University, Raleigh, NC 27606, USA}

\author{Yuan Liu}
\affiliation{Department of Electrical and Computer Engineering, North Carolina State University, Raleigh, NC 27606, USA}
\affiliation{Department of Computer Science, North Carolina State University, Raleigh, NC 27606, USA}
\affiliation{Department of Physics, North Carolina State University, Raleigh, NC 27606, USA}
\email{q\_yuanliu@ncsu.edu}

\maketitle

\begin{abstract}
We determine the non-Clifford $T$-gate cost of constructing block
encodings of structured fermionic and spin Hamiltonians in a unitary Clifford$+T$ model, when arbitrarily many clean
ancillas and unrestricted block-encoding subnormalization are allowed, but without mid-circuit
measurements or classical feed-forward. Our main technical tool is an
ancilla-compression theorem: any block encoding of an $n$-qubit operator with $a$ clean ancillas and at most $s$ $T$ gates can be compressed to use at most $\min\{a,n+2s\}$ ancillas, without increasing the absolute error or $T$-count.
For general second-quantized Hamiltonians with bounded one- and
two-body coefficients, at operator-norm block-encoding error $\epsilon$, a volume-covering argument combined with
circuit counting gives the worst-case lower bound
$\Omega\bigl(n^2\sqrt{\log(n^4/\eps)}\bigr)$, matching the existing
upper bound at fixed precision. For the bond-dependent Kitaev honeycomb family on $n$ spins, we obtain
independent lower bounds $\Omega(n)$ from stabilizer nullity and
$\Omega(\log(1/\eps))$ from one-qubit state preparation, established using different Hamiltonian instances. Together with an explicit LCU construction, they
give the tight worst-case scaling
$\Theta(n+\log(1/\eps))$.
As an application, we evaluate the $T$-count of a
Hamiltonian simulation circuit based on quantum singular value transformation, with each block-encoding query compiled separately. When phase synthesis and controlled queries add at most constant-factor overhead, the simulation $T$-count scales as the query count times the optimal $T$-count per query.
\end{abstract}

\section{Introduction}

Block encoding provides a general framework for representing a nonunitary
operator as a subblock of a larger unitary transformation.  It has become a
standard input model for quantum algorithms based on qubitization
\cite{lowchuang2019qubitization} and quantum singular value transformation (QSVT)
\cite{gilyen2019qsvt}, including algorithms for Hamiltonian simulation and
quantum linear algebra \cite{10.1063/5.0312254,zlokapa2024hamiltonian,chakraborty2019power}. The cost of constructing a block encoding can
be a significant part of the total cost of such algorithms.

In fault-tolerant quantum computation it is useful to distinguish the cost of
Clifford operations from non-Clifford operations.  In the
Clifford$+T$ model, Clifford operations are comparatively inexpensive, whereas
$T$ gates require non-Clifford resources, typically supplied through
magic-state distillation and injection \cite{bravyi2005magic}.  The
$T$-count---the number of $T$ gates in a circuit---is therefore a natural
measure of non-Clifford cost, and a growing body of work has investigated the
optimal $T$-count of basic quantum algorithms
\cite{gheorghiu2022tcount,low2024trading,gosset2026stateprep,li2026sparsity}. Earlier work developed
Clifford$+T$ synthesis and optimization methods, including $T$-count and
$T$-depth optimization and near-optimal approximation of single-qubit rotations
\cite{amy2014tdepth,selinger2015efficient,ross2016optimal}. Low \emph{et al.} developed space--$T$-count tradeoffs for data
lookup, state preparation, and unitary synthesis \cite{low2024trading}.
Gosset \emph{et al.} subsequently determined the optimal asymptotic
$T$-count for arbitrary state preparation and diagonal unitary synthesis
\cite{gosset2026stateprep}.  More recently, Li \emph{et al.} 
obtained optimal
$T$-count bounds for sparse QROM, sparse state preparation, and block encodings
of general sparse matrices \cite{li2026sparsity}.

For Hamiltonian simulation, however, the matrices to be block encoded are
typically not arbitrary.  Their entries are constrained by the algebraic and
geometric structure of the underlying physical model, and these constraints may
substantially reduce the resources required.  This raises a different question:

\begin{quote}
\emph{What is the asymptotically optimal $T$-count required, in the
worst case, to construct an $\eps$-approximate block encoding of a
Hamiltonian from a structured family $\mathfrak H_m$?}
\end{quote}

In this paper we study one fermionic and one spin family. The first consists of general
second-quantized Hamiltonians with bounded one- and two-body
coefficients. For this family, Liu \emph{et al.} constructed block
encodings using $    O\!\left(n^{2}\sqrt{\log(n^{4}/\eps)}\right) $
$T$ gates at fixed precision
\cite{liu2025block}. Here $n$ is the number of fermionic modes,
equivalently the number of system qubits, and $\eps$ is the
block-encoding error. The second consists of bond-dependent Kitaev honeycomb spin Hamiltonians on $n$ spins,
with real bond couplings $J_e\in[-1,1]$
\cite{kitaev2006anyons,yamada2020anderson}. For this family, we construct
an explicit LCU block encoding using
$
    O\!\left(n+\log(1/\eps)\right)
$
$T$ gates. We ask whether either upper bound can be improved
asymptotically in the worst case.

\subsection{Summary of results and scope}
We work in the unitary Clifford$+T$ circuit model. We allow arbitrarily
many clean ancillas to fully explore tradeoff between circuit width and depth, and place no restriction on the block-encoding
subnormalization factor $\alpha>0$. Mid-circuit measurements and
classical feed-forward are excluded. All complexity bounds below are
worst-case over the corresponding Hamiltonian family.

\paragraph{Main result 1.}
Our first result concerns the family $\cH_n$ of 
general
bounded-coefficient $n$-fermionic mode Hamiltonians in second quantization. 
Its one- and two-body coefficients obey
$h_{qp}=\overline{h_{pq}}$, $V_{RP}=\overline{V_{PR}}$,
$|h_{pq}|\le1$, and $|V_{PR}|\le1$.
These are the only restrictions on the coefficients; in particular, we do not impose the additional relations arising from molecular Coulomb
integrals. The resulting general algebraic family is defined in
Section~\ref{sec:secondq}.
Under the conditions of Theorem~\ref{thm:secondq-main},
\begin{equation}
    T^{\BE}_{\eps}\!\left(\cH_n\right)
    =\Theta\!\left(
      n^{2}\sqrt{\log\!\left(\frac{n^{4}}{\eps}\right)}
    \right).
    \label{eq:intro-secondq}
\end{equation}
The lower bound follows from the circuit-counting argument in
Section~\ref{sec:circuit-counting}. Combined with the construction of
Liu \emph{et al.}~\cite{liu2025block}, it proves that this scaling is
asymptotically optimal for the general algebraic family.

\paragraph{Main result 2.}
Our second result concerns the bond-dependent Kitaev honeycomb spin family
$\cK_n$ on $n$ spins, defined in Section~\ref{sec:kitaev}. Its real bond
couplings satisfy $J_e\in[-1,1]$. Under the conditions of Theorem~\ref{thm:kitaev-optimal},
\begin{equation}
    T^{\BE}_{\eps}(\cK_n)
    =\Theta\!\left(n+\log\frac{1}{\eps}\right).
    \label{eq:intro-kitaev}
\end{equation}
An explicit LCU construction
developed in this work gives the matching upper bound
$O(n+\log(1/\eps))$.

\paragraph{Main result 3.}
Finally, we apply our results to a fixed QSVT circuit
for Hamiltonian simulation. Let $N_{\rm calls}$ count all uncontrolled block-encoding queries,
including those used in OAA, and let $B$ be the optimal $T$-count
per query at a fixed normalization. We compile each query as a
separate module. If the phase-synthesis and controlled-query costs
are both $O(N_{\rm calls}B)$, the total modular $T$-count is
$\Theta(N_{\rm calls}B)$. This is a bound for the fixed
QSVT circuit as written directly from the algorithms before additional circuit compilation, not for all Hamiltonian simulation algorithms.

\paragraph{Scope and limitations.}
Our results concern the two Hamiltonian families defined above and
unitary Clifford$+T$ circuits without mid-circuit measurements or
classical feed-forward. The fermionic lower bound is worst-case over
the full bounded-coefficient family, and does not apply to every
Hamiltonian. This means that for molecular Hamiltonians with more structures in their coefficients, there may be even smaller lower bound on the T-count than what we established here. For the
Kitaev family, the $\Omega(n)$ and $\Omega(\log(1/\eps))$ bounds follow from two different Hamiltonian instances. The QSVT result uses modular
compilation of a fixed circuit and is not a lower bound for all
Hamiltonian simulation algorithms. 

\paragraph{Ancilla compression.}
A direct circuit count fails when the number of clean ancillas is
unrestricted. Our compression theorem removes this obstacle. Given a
block encoding of an $n$-qubit operator using at most $s$ $T$ gates,
there exists a new block encoding derived from the original one. It acts
on at most $2(n+s)$ qubits, uses no more $T$ gates, and achieves the same
absolute error. The relevant circuit family is therefore finite, and the
second-quantized lower bound no longer depends on the original ancilla
count. The proof combines symplectic reduction for non-Clifford Pauli
rotations with a structural analysis of postselected stabilizer maps.
\paragraph{Proof overview.}
The two lower bounds proceed by different mechanisms.  We describe them
informally here; the objects they refer to are defined in
Sections~\ref{sec:secondq} and~\ref{sec:kitaev}.

For the general second-quantized Hamiltonian family the proof is a volume-counting argument in the
space of Hamiltonian coefficients.  Fix a Clifford$+T$ circuit and look only at
the operator obtained by projecting its ancillas onto $\ket{0^a}$ for $a$ being the number of ancilla qubits.  Changing the
normalization of the block encoding only rescales this projected operator, so
this circuit can approximate Hamiltonians only near a single line segment
in coefficient space.  The allowed Hamiltonian family, however, contains a
full-dimensional body of coefficients.  Hence many different circuits are
needed to cover the family at accuracy $\eps$.  
Combining this covering requirement with ancilla compression and
circuit counting yields the desired $T$-count lower bound.

For the spin family, the Kitaev bounds require different arguments, since that family has
only $\Theta(n)$ unique coupling coefficients and the bound from the counting method is
correspondingly weak.  For the system-size dependence we choose a stabilizer
input on which the Hamiltonian acts to produce a state of maximal stabilizer
nullity, and invoke the monotonicity of nullity under stabilizer operations
\cite{beverland2020lower}.  For the precision dependence we restrict to a
two-bond subfamily and reduce block encoding to the preparation of a one-qubit
state requiring $\Omega(\log(1/\eps))$ non-Clifford
resources.  An explicit PREPARE--SELECT construction supplies the matching
upper bound.

The QSVT application follows by combining the two block-encoding bounds
with the standard query complexity of QSVT Hamiltonian simulation. The
additional gate and error estimates are given in
Appendix~\ref{app:qsp-cost}.

\section{Framework and notation}\label{sec:framework}

\subsection{Block encodings and worst-case \texorpdfstring{$T$}{T} complexity}

For an $n$-qubit operator $H$, let $U$ be a unitary on $n+a$ qubits and define
\begin{equation}\label{eq:corner-def}
  A_U:=(\bra{0^a}\otimes I)\,U\,(\ket{0^a}\otimes I).
\end{equation}
We call $U$ an \emph{$(\alpha,a,\eps)$ block encoding} of $H$ if
\begin{equation}\label{eq:block-def}
  \|H-\alpha A_U\|\le\eps,\qquad\alpha>0,
\end{equation}
where $\|\cdot\|$ denotes the operator norm. Throughout, unless stated otherwise, all block-encoding error parameters are assumed to satisfy
$0<\eps\le1$. The ancilla register may have
arbitrary size, and we only require $\alpha>0$.  We do not assume
$\alpha\ge\|H\|$.  In the exact case $\eps=0$, this follows from
$\|A_U\|\le1$.  For nonzero error, we can only conclude
$\alpha\ge\|H\|-\eps$.
We use this weaker definition because the compression in
Theorem~\ref{thm:compression} may decrease $\alpha$.

All lower-bound arguments use unitary Clifford$+T$ circuits, with no
intermediate measurements or classically controlled gates.  For counting
purposes, each $T^\dagger$ can be replaced by $S^\dagger T$, where
$S^\dagger$ is Clifford.  We therefore use $T(U)$ for the number of $T$
gates after this replacement.
For a family $\mathfrak H$ of Hamiltonians, the worst-case
block-encoding $T$ complexity is
\begin{equation}\label{eq:worstcase}
  T^{\BE}_{\eps}(\mathfrak H)
  :=\sup_{H\in\mathfrak H}\ \
    \inf_{\substack{a\ge0,\ \alpha>0,\ U\\ \|H-\alpha A_U\|\le\eps}}T(U).
\end{equation}
Here the inner infimum is over all valid block encodings of a fixed $H$ with
error at most $\eps$ and arbitrary $a \geq 0$ and $\alpha > 0$, while the outer supremum takes the worst case over
$H\in\mathfrak H$.
All lower bounds below are worst-case statements in this sense.

For the QSVT application, we also need to keep the normalization fixed,
because the number of QSVT queries depends on $\alpha$. We define
\[
T^{\BE}_{\eps,\alpha}(\mathfrak H)
:=
\sup_{H\in\mathfrak H}
\inf_{\substack{a,U:\\
\|H-\alpha A_U\|\le\eps}}
T(U).
\]
Since the normalization is fixed,
\[
T^{\BE}_{\eps,\alpha}(\mathfrak H)
\ge
T^{\BE}_{\eps}(\mathfrak H).
\]
We use this fixed-normalization quantity only in
Section~\ref{sec:qsp-application}.

\subsection{Canonical form of Clifford\texorpdfstring{$+T$}{+T} circuits}
We use the following canonical form for unitary Clifford$+T$ circuits, based
on Gosset \emph{et al.}~\cite[Proposition~1]{gosset2014tcount}.  If a circuit $U$ satisfies
$T(U) \leq s$, then, up to a global phase,
it can be written as
\begin{equation}\label{eq:pauli-normal-form}
    U=C\prod_{j=1}^{s}R(P_j),
    \qquad
    R(P):=\exp\!\left(-\frac{i\pi}{8}P\right),
\end{equation}
where $C$ is Clifford and each $P_j$ is a Hermitian generator (for example, a Pauli string).  We allow
$P_j=\pm I$ in this representation: such factors contribute only a global
phase. They also serve to pad a representation containing fewer than $s$
nontrivial Pauli rotations to exactly $s$ factors if $T(U) < s$.
Since multiplying a Clifford unitary by a global phase is still a Clifford, we absorb the overall phase into $C$ and hence use
\eqref{eq:pauli-normal-form} as an exact equality.

\subsection{Binary symplectic representation}\label{subsec:symplectic}

Here we briefly introduce the binary symplectic representation needed for
the proof of the ancilla-compression theorem, Theorem~\ref{thm:compression}.
For further background on symplectic vector spaces and symplectic groups,
see, e.g., Ref.~\cite{artin1957geometric,grove2002classical}; the specific facts used here are
collected in Appendix~\ref{app:symplectic}.

Ignoring phases, the $a$-qubit Pauli group can be identified with
$V=\F_2^{2a}$ through the mapping
\begin{equation}\label{eq:nu}
\nu\bigl(X^{x}Z^{z}\bigr)=(x,z).
\end{equation}
The symplectic form $\omega$ on $V$ is defined by
\[
\omega\bigl((x,z),(x',z')\bigr)
=x\cdot z'+z\cdot x'\pmod2,
\]
so that two Pauli operators commute if and only if their corresponding
symplectic product is $0$.  We write
$e_i=(\mathbf e_i,0)$ and $f_i=(0,\mathbf e_i)$ for the standard symplectic
basis, and
\begin{equation}\label{eq:LZ}
L_Z=\{(0,z):z\in\F_2^{a}\}
=\Span\{f_1,\ldots,f_a\}
\end{equation}
for the subspace which is the image of the $Z$-type Paulis under $\nu$.
Furthermore, conjugation by Clifford operators will induce symplectic transformations of
$V$ \cite{dehaene2003clifford}.

\section{Ancilla compression}
\label{sec:compression}

Circuit counting requires a bound on the total number of qubits.
However, the definition of worst-case block-encoding complexity allows
an unrestricted number of clean ancillas. We therefore need an
ancilla bound that depends only on the system size and the $T$-count.

Allowing arbitrarily many clean ancillas is a relaxation of the
circuit model, not an assumption about practical workspace. In many
practical constructions the ancilla count satisfies $a\le s$. Then
$n+a\le n+s$, so the original circuit already satisfies the required
qubit bound and no compression is needed. The compression theorem
rules out the possibility that a circuit could evade the counting
lower bound by using an asymptotically larger ancilla register.

Beverland \emph{et al.}
eliminate stabilizer ancillas from postselected stabilizer
computations~\cite[Theorem~5.3]{beverland2020lower}. Related ancilla-independent canonical forms have also
been used in counting lower bounds for state preparation and block encodings of sparse matrices \cite{li2026sparsity}.

The result needed here is not a direct application of these earlier
compression statements. We use the unitary block-encoding structure
and treat the non-Clifford and Clifford parts of the circuit
separately. The Pauli rotations are first compressed to a small register.
The remaining projected Clifford block is then characterized through
its Choi state and the bipartite stabilizer normal form of
Fattal \emph{et al.}~\cite{fattal2004entanglement}. This produces a
unitary block encoding with a bounded number of ancillas.

Our construction also preserves the absolute approximation error and
does not increase either the $T$-count or the normalization. It
therefore allows us to restrict the subsequent counting argument to a
finite number of qubits without weakening the block encoding.

\begin{theorem}[Ancilla compression]
\label{thm:compression}
Let $H$ act on $n$ qubits, and let $U$ be an
$(\alpha,a,\eps)$ block encoding of $H$ with $T(U)\le s$.
Then there exist an integer $\kappa\ge0$ and a unitary circuit $U'$
that is an $(\alpha',a',\eps)$ block encoding of $H$, where
\[
    \alpha'=2^{-\kappa/2}\alpha,
    \qquad
    T(U')\le s,
    \qquad
    a'\le\min\{a,n+2s\}.
\]
Equivalently,
\[
    n+a'\le\min\{n+a,2(n+s)\}.
\]
In particular, the ancilla count is strictly reduced whenever
$a>n+2s$, and $\alpha'\le\alpha$.
\end{theorem}

\begin{proof}[Proof sketch]
Let $\widetilde H=\alpha A_U$ be the operator encoded exactly by $U$.
Then $\|H-\widetilde H\|\le\eps$, so it is enough to compress the
exact block encoding of $\widetilde H$.

Write $U$ as a Clifford followed by at most $s$ Pauli $\pi/8$
rotations. Proposition~\ref{prop:pauli-compression} gives an ancilla
Clifford that fixes $\ket{0^a}$ and confines the ancilla support of
all these rotations to at most $s$ qubits. After projecting out the
remaining ancillas, the Clifford part gives a projected block
$\Lambda$. By Theorem~\ref{thm:corner-structure},
\[
    \Lambda=2^{-\kappa/2}W\Pi,
\]
where $W$ is Clifford and $\Pi$ is a stabilizer-code projector with at
most $n+s$ independent generators.

Proposition~\ref{prop:reconstruct} implements $\Pi$ coherently using
at most $n+s$ syndrome qubits and Clifford gates only. The resulting
circuit therefore uses at most
\[
    s+(n+s)=n+2s
\]
ancillas, has $T$-count at most $s$, and has normalization
$2^{-\kappa/2}\alpha$. If the original circuit already uses at most
$n+2s$ ancillas, we retain it instead. The complete construction,
including the zero-block case, is given in
Appendices~\ref{app:pauli-compression}--\ref{app:compressed-circuit}.
\end{proof}

\section{General second-quantized Hamiltonians}\label{sec:secondq}

\subsection{Hamiltonian family and coefficient geometry}

Let $a_1,\ldots,a_n$ be fermionic annihilation operators.  Write
$\cI=\{(p,q):1\le p<q\le n\}$ and $D=|\cI|=\binom n2$, and for $P=(p,q)\in\cI$
put
\[
  b_P^{\dagger}=a_p^{\dagger}a_q^{\dagger},\qquad b_P=a_qa_p.
\]
We consider the family of number-conserving one- and two-body Hamiltonians
\begin{equation}\label{eq:secondq-H}
  H(h,V)=\sum_{p,q=1}^{n}h_{pq}a_p^{\dagger}a_q
         +\sum_{P,R\in\cI}V_{PR}\,b_P^{\dagger}b_R,
\end{equation}
where $h=(h_{pq})_{1\le p,q\le n}\in\C^{n\times n}$ is the one-body
coefficient matrix and $V=(V_{PR})_{P,R\in\cI}\in\C^{D\times D}$ is the
two-body coefficient matrix.  The coefficients satisfy the Hermiticity
constraints $h_{qp}=\overline{h_{pq}}$ and $V_{RP}=\overline{V_{PR}}$.
We consider an algebraic family in which $h$ and $V$ vary independently.
They are constrained only by Hermiticity and the entry-wise bounds stated below.  In particular, \eqref{eq:secondq-H} includes molecular electronic Hamiltonians in which their integral
coefficients arise from one- and two-electron integrals of a Coulomb
interaction in a chosen orbital basis, although they may obey
additional structural constraints among their coefficients arising from molecular geometries.  Accordingly,
the $\Theta(n^4)$ parameter count below applies to the general algebraic
family considered here; we make no claim that the Coulomb-integral molecular
subfamily has the same parameter dimension.
The independent real parameters are the diagonal entries of $h$ and $V$
together with the real and imaginary parts of their strictly upper-triangular
entries, so the real parameter dimension is
\begin{equation}\label{eq:d-secondq}
  d=n^{2}+D^{2}=\frac{n^{4}-2n^{3}+5n^{2}}{4}=\Theta(n^{4}).
\end{equation}

We now work with the bounded-coefficient family $\cH_n$: throughout this
section, every Hamiltonian in the family is assumed to satisfy
\begin{equation}
  |h_{pq}|\le 1 \quad (1\le p,q\le n),\qquad
  |V_{PR}|\le 1 \quad (P,R\in\cI).
\end{equation}
Identify a pair $(h,V)$ with its
vector $x\in\R^{d}$ of independent real coordinates and let $\cK\subset\R^{d}$
be the corresponding coefficient body.  Each complex unit disc contains a
centred square of side $\sqrt2$ and each real diagonal coordinate ranges over
$[-1,1]$, so
\begin{equation}\label{eq:K-volume}
  \vol(\cK)\ge 2^{d/2},
  \qquad
  \|x\|_2\le\sqrt d\quad(x\in\cK),
\end{equation}
the second inequality follows because every real coordinate has modulus at most one.
The map $x\mapsto H(x)$ is $\R$-linear, and $\cK$ is compact, so
\begin{equation}\label{eq:Gamma}
  \Gamma:=\sup_{x\in\cK}\|H(x)\|<\infty .
\end{equation}
The following lemma relates the two norms used in the proof.
Block-encoding error is measured in operator norm.
The volume argument instead uses the Euclidean norm of the coefficient vector.  Unlike the coefficient body
$\cK$, this estimate is homogeneous and applies to arbitrary coefficients, not
only to those satisfying the entrywise bounds above.

\begin{lemma}\label{lem:coefficient-norm}
For every $n\ge2$ and every coefficient vector $x=(h,V)$,
\begin{equation}\label{eq:norm-comparison}
  \|H(h,V)\|\ \ge\ \frac{1}{4n^{3/2}}\,\|x\|_2 .
\end{equation}
\end{lemma}
The proof is given in Appendix~\ref{app:second-coefficient-norm}.

\subsection{The set covered by one circuit}

Fix a unitary $U$ together with a division of its qubits into $a$ block
ancillas and $n$ system qubits, and let $[U]$ denote its equivalence class
up to global phase. Define
\begin{equation}\label{eq:SU}
  S_{[U]}
  :=
  \bigl\{
    x\in\cK:\ \exists\,z\in\C\setminus\{0\}\ \text{with}\
    \|H(x)-zA_U\|\le\eps
  \bigr\}.
\end{equation}
This definition is independent of the representative chosen from $[U]$.

\begin{lemma}[One circuit covers a tube]\label{lem:tube}
Let $c_0=(4n^{3/2})^{-1}$ and $\rho=4\eps/c_0=16n^{3/2}\eps$.  If $S_{[U]}\ne\varnothing$ then
there exists $x_\star\in\cK$ with
\begin{equation}\label{eq:tube-inclusion}
  S_{[U]}\subseteq\bigl\{x\in\R^{d}:\dist\bigl(x,\{\mu x_\star:-2\le\mu\le2\}\bigr)\le\rho\bigr\},
\end{equation}
and consequently
\begin{equation}\label{eq:tube-volume}
  \vol(S_{[U]})\ \le\ (4\sqrt d+3\rho)\,\omega_{d-1}\rho^{\,d-1},
\end{equation}
where $\omega_k$ is the volume of the $k$-dimensional Euclidean unit ball.
\end{lemma}

\begin{proposition}[Covering requirement]\label{prop:covering}
Let $c_\star:=\bigl(32\,\sqrt{2\pi e}\bigr)^{-1}$ and assume
$0<\eps\le1$ and $n\ge2$.  If a collection of $M$ different unitary classes up to a global phase has encoded
regions $S_{[U]}$ whose union covers $\cK$, then
\begin{equation}\label{eq:covering-count}
  \log_2 M\ \ge\ (d-1)\,\mathcal L\ -\ 4\log_2 d,
  \qquad
  \mathcal L:=\log_2\frac{c_\star\,\sqrt n}{\eps}.
\end{equation}
\end{proposition}
The proofs of Lemma~\ref{lem:tube} and Proposition~\ref{prop:covering}
are given in Appendix~\ref{app:second-covering-est}.

\begin{remark}\label{rem:L-equivalence}
For $0<\eps\le1$ one has
$\mathcal L=\tfrac12\log_2 n+\log_2(1/\eps)-\log_2(1/c_\star)$ and
$\log_2(n^{4}/\eps)=4\log_2 n+\log_2(1/\eps)$, so
\[
  \mathcal L\ \ge\ \tfrac18\log_2\frac{n^{4}}{\eps}-\log_2\frac1{c_\star},
  \qquad
  \mathcal L\ \le\ \log_2\frac{n^{4}}{\eps}.
\]
Hence $\mathcal L=\Theta\bigl(\log(n^{4}/\eps)\bigr)$ whenever
$\log_2(n^{4}/\eps)\ge 8\log_2(1/c_\star)$, which holds for all $n$ beyond an
absolute constant.  We use this equivalence without further comment.
\end{remark}

\subsection{Circuit counting and the lower bound}
\label{sec:circuit-counting}
\begin{lemma}[Circuit count]\label{lem:circuit-count}
Let $\Ncir(Q,s)$ be the number of distinct unitaries on $Q$ qubits, up to
global phase, implemented by a unitary Clifford$+T$ circuit with at most $s$
$T$ gates.  Then
\begin{equation}\label{eq:circuit-count}
  \log_2\Ncir(Q,s)\ \le\ 2Q^{2}+3Q+s(2Q+1).
\end{equation}
\end{lemma}

\begin{proof}
Every unitary implemented with at most $s$ $T$ gates can be written in the form
\eqref{eq:pauli-normal-form} with exactly $s$ factors once trivial factors
$P_j=\pm I$ are permitted, since such a factor contributes only a global phase.
Up to global phase, $U$ is therefore determined by the Clifford $C$ together
with the ordered tuple $(P_1,\ldots,P_s)$ of Hermitian Paulis.  The number of
Hermitian Paulis on $Q$ qubits, including $\pm I$, is $2\cdot4^{Q}=2^{2Q+1}$.
The order of the $Q$-qubit Clifford group modulo phases is
\[
  |\Cl_Q|=2^{Q^{2}+2Q}\prod_{j=1}^{Q}(4^{j}-1)
  <2^{Q^{2}+2Q}\cdot 2^{Q(Q+1)}=2^{2Q^{2}+3Q}.
\]
Multiplying these choices upper bounds the number of distinct unitaries
implemented by such Clifford$+T$ circuits, since different choices may represent
the same unitary.
\end{proof}

\begin{theorem}[Second-quantized lower bound]\label{thm:secondq-lower}
Let $n\ge16$ and $0<\eps\le \min\{1,c_\star\sqrt n/2\}$.  Then
\begin{equation}\label{eq:secondq-lb}
  T^{\BE}_{\eps}\bigl(\cH_n\bigr)
  \ =\ \Omega\!\left(n^{2}\sqrt{\log(n^{4}/\eps)}\right).
\end{equation}
The bound is independent of the ancilla count and of the subnormalization.
\end{theorem}

\begin{proof}
Suppose every $H\in\cH_n$ admits an $\eps$ block encoding by a unitary
Clifford$+T$ circuit with at most $s$ $T$ gates, with arbitrary ancilla count
and arbitrary subnormalization.  By Theorem~\ref{thm:compression} each such
encoding may be replaced by a block encoding of the same target, with the same $\eps$, with
$T$-count at most $s$, and on at most
\[
  Q:=2(n+s)
\]
total qubits.  Padding with idle clean qubits makes every circuit act on
exactly $Q$ qubits without changing its encoded block.  Hence the $\Ncir(Q,s)$  encoded regions
$S_{[U]}$ cover $\cK$, and Proposition~\ref{prop:covering} together with
Lemma~\ref{lem:circuit-count} gives
\[
  2Q^{2}+3Q+s(2Q+1)\ \ge\ (d-1)\mathcal L-4\log_2 d.
\]
Substituting $Q=2n+2s$ and expanding,
\[
  8(n+s)^{2}+6(n+s)+s(4n+4s+1)
  =
  12s^{2}+(20n+7)s+(8n^{2}+6n)
  \ge (d-1)\mathcal L-4\log_2 d.
\]
Solving this quadratic inequality for $s$ gives
\begin{equation}\label{eq:secondq-quadratic-root}
  s\ \ge\
  \frac{-(20n+7)+
  \sqrt{48\bigl((d-1)\mathcal L-4\log_2 d\bigr)+16n^{2}-8n+49}}{24}.
\end{equation}
Dropping the nonnegative term $16n^{2}-8n+49$ under the square root in
\eqref{eq:secondq-quadratic-root} only decreases the right-hand side, so
\[
  s\ \ge\
  \frac{\sqrt{48\bigl((d-1)\mathcal L-4\log_2 d\bigr)}-(20n+7)}{24}.
\]
For the square-root term, Remark~\ref{rem:L-equivalence} and $d=\Theta(n^4)$
imply
\[
  \sqrt{48\bigl((d-1)\mathcal L-4\log_2 d\bigr)}
  =\Theta\!\left(n^2\sqrt{\log(n^4/\eps)}\right).
\]
Since $n=O(n^2\sqrt{\log(n^4/\eps)})$ in the stated range, subtracting
$(20n+7)/24$ does not change the lower-bound order.
Thus $T^{\BE}_{\eps}\bigl(\cH_n\bigr)
  \ =\ \Omega\!\left(n^{2}\sqrt{\log(n^{4}/\eps)}\right)$.
\end{proof}

Liu \emph{et al.}~\cite{liu2025block} construct block encodings for general second-quantized
Hamiltonians with $L=O(n^{4})$ interaction coefficients and optimize a
SELECT--SWAP lookup tradeoff to obtain
\begin{equation}\label{eq:liu-upper}
  O\!\left(n+\sqrt{L\,m_b}\right)=O\!\left(n^{2}\sqrt{\log(n^{4}/\eps)}\right)
\end{equation}
$T$ gates in the regime $\log(n^{4}/\epsilon)=O(n^{4})$, with
$m_b=\Theta(\log(n^{4}/\eps))$.; the additive
$O(n)$ is absorbed by the displayed $O(n^2\sqrt{\log(n^4/\eps)})$ bound.
This upper bound applies to the present family because its number of interaction
terms is $O(n^4)$, while the lower bound above uses the stronger fact that the
coefficient space has dimension $d=\Theta(n^4)$.

For reference, we provide the T-count of a set of known molecules in second quantization, and compare the construction in Ref. \cite{PhysRevX.8.041015} to our lower bound in Fig \ref{fig:distribution}. The molecules are extracted using the software PySCF in near equilibrium positions under the STO-3G basis set. Then a Jordan Wigner transformation is implemented to get the Pauli string corresponding to the description of the Hamiltonian which are the number of terms to block encode.

\begin{figure}[h]
    \centering
    \includegraphics[width=0.8\textwidth]{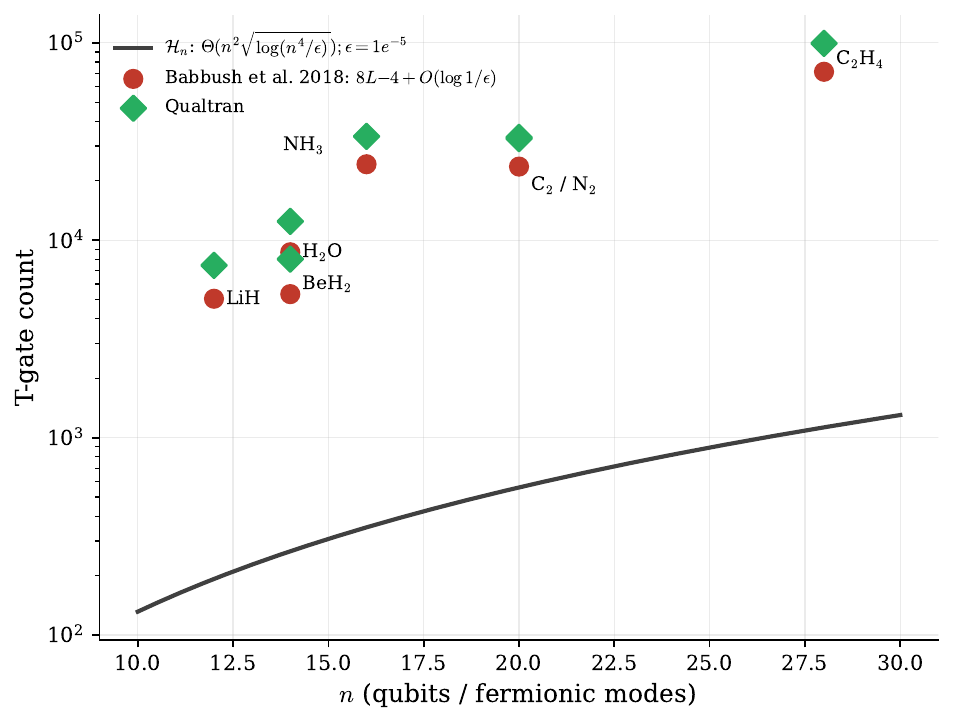}
    \caption{T-gate counts for different chemistry Hamiltonian in second quantization with STO-3G basis, comparing out lower bound (black line) with prior constructions in Ref. \cite{PhysRevX.8.041015} and Qualtran \cite{harrigan2024qualtran}. The constant factor $\sqrt{48}/24 \approx 0.288$ is included for our lower bound. $\epsilon=10^{-5}$ is chosen for all cases.}
    \label{fig:distribution}
\end{figure}

\begin{theorem}[Optimal scaling for the general family]\label{thm:secondq-main}
Let $\cH_n$ be the family defined in Eq.~\eqref{eq:secondq-H} of second-quantized
Hamiltonians $H(h,V)$ with $|h_{pq}|\le1$ for all
$1\le p,q\le n$ and $|V_{PR}|\le1$ for all $P,R\in\cI$. For every fixed $0<\epsilon \leq 1$, and for all sufficiently large $n$,
\begin{equation}\label{eq:secondq-theta}
  T^{\BE}_{\eps}\bigl(\cH_n\bigr)=\Theta\!\left(n^{2}\sqrt{\log(n^{4}/\eps)}\right)
\end{equation}
in the unitary Clifford$+T$ circuit model of Section~\ref{sec:framework}.
\end{theorem}

\begin{proof}
Combine \eqref{eq:secondq-lb} with \eqref{eq:liu-upper}.
\end{proof}

\begin{corollary}[Fixed normalization]
\label{cor:secondq-fixed-normalization}
There is a normalization
\[
  \overline\alpha_{2,n}=\Theta(n^4)
\]
that is independent of the coefficients $h$ and $V$.
For every fixed $0<\eps\le1$, and for all sufficiently large $n$,
\[
  T^{\BE}_{\eps,\overline\alpha_{2,n}}
  \bigl(\cH_n\bigr)
  =
  \Theta\!\left(
    n^2\sqrt{\log(n^4/\eps)}
  \right).
\]
The matching construction uses
\[
  O\!\left(
    n^2\sqrt{\log(n^4/\eps)}
  \right)
\]
ancillas.
\end{corollary}

\begin{proof}
The lower bound follows from the unrestricted-normalization result.
The padded version of the construction of Liu \emph{et al.} gives the
upper bound \cite{liu2025block}; see
Appendix~\ref{app:secondq-common-normalization}.
\end{proof}

\section{Bond-dependent Kitaev honeycomb Hamiltonians}\label{sec:kitaev}

We now consider a family of Pauli spin Hamiltonians based on the
Kitaev honeycomb interaction pattern. Let $G=(V,E)$ be a honeycomb graph with $n=|V|\ge4$ vertices and bipartition
$V=A\sqcup B$, whose edge set decomposes into three perfect matchings
$E=E_x\sqcup E_y\sqcup E_z$; thus every vertex has exactly one neighbour of
each bond type, and $|E_x|=|E_y|=|E_z|=n/2$.  For an edge $e=(u,v)$ set
\[
  P_e=\begin{cases}
    X_uX_v,&e\in E_x,\\
    Y_uY_v,&e\in E_y,\\
    Z_uZ_v,&e\in E_z,
  \end{cases}
\]
following the bond-disordered variant of the Kitaev honeycomb interaction
pattern \cite{kitaev2006anyons,yamada2020anderson}. And we define the family
\begin{equation}\label{eq:kitaev-H}
  \cK_n=\Bigl\{H_K(J)=\sum_{e\in E}J_eP_e:\ J_e\in[-1,1]\Bigr\},
  \qquad M:=|E|=\tfrac32n .
\end{equation}

We prove that, for the family \eqref{eq:kitaev-H},
\[
  T^{\BE}_{\eps}(\cK_n)=\Theta\!\left(n+\log\frac1\eps\right)
\]
in the parameter range stated in Theorem~\ref{thm:kitaev-optimal}.

The coefficient-space covering argument of Section~\ref{sec:secondq} is not the
right tool here: it is tailored to the full second-quantized coefficient body,
where the number of independent parameters is quartic in the system size,
whereas \eqref{eq:kitaev-H} has only $M=\Theta(n)$ couplings and a different
Pauli structure.  Even if one had an analogous coefficient-space estimate with
only $M$ parameters, the resulting comparison with the circuit count would have
the schematic form
\[
  12s^{2}+O(ns)+O(n^{2})\ \gtrsim\ M\,\mathcal L .
\]
Since $M=\Theta(n)$ and $\mathcal L$ is logarithmic in $n$ and $1/\eps$,
solving such an inequality gives a root of the form
\[
  s\ \gtrsim\ -O(n)+\sqrt{O(n^{2})+O(n\mathcal L)} ,
\]
and the negative linear term cannot be controlled by the square-root term in
the parameter range relevant here.  Thus this route does not produce the desired
linear lower bound.  We instead use two model-specific arguments: a
stabilizer-nullity bound (Section~\ref{subsec:kitaev-nullity}) for the linear
term, and a reduction to one-qubit state preparation
(Section~\ref{subsec:kitaev-precision}) for the logarithmic term.

\subsection{A linear lower bound from stabilizer nullity}\label{subsec:kitaev-nullity}

Stabilizer nullity is especially useful for proving linear lower bounds on
non-Clifford resources: it was introduced and used by Beverland
\emph{et al.}~\cite{beverland2020lower} to lower bound resource conversions
and to prove linear lower bounds for several gates, e.g.\ multiply
controlled $Z$ gates.  The same technique for proving that exact multi-qubit
Toffoli gates have linear complexity is also a starting point in the recent
work of Gosset \emph{et al.} on approximate Toffoli
implementations~\cite{gosset2025toffoli}.  We show that this monotone is
effective for the Kitaev family: for suitable couplings, applying $H_K(J)$ to a
stabilizer input produces a state with maximal stabilizer nullity.

For a $Q$-qubit state $\sigma$, define the stabilizer nullity
\begin{equation}\label{eq:nullity}
  \nu(\sigma)
  =
  Q-\log_2\Bigl|\bigl\{
  P\in\{\pm1\}\cdot\{I,X,Y,Z\}^{Q}:P\sigma=\sigma
  \bigr\}\Bigr| .
\end{equation}
For a pure state we write $\nu(\ket\psi):=\nu(\proj\psi)$.
We use the following standard properties \cite{beverland2020lower}: $\nu$
vanishes on stabilizer states, is invariant under Clifford unitaries, is
additive under tensoring with a stabilizer state, does not increase under
postselected Pauli measurement, and increases by at most one under a $T$ or
$T^{\dagger}$ gate.  Consequently a state obtained from a stabilizer input by a
unitary Clifford$+T$ circuit with $s$ $T$ gates followed by Pauli postselection
has nullity at most $s$.

Let the local bases on the two sublattices be
\[
  \ket{0_A}=\frac{\ket{0_c}+\ket{1_c}}{\sqrt2},\quad
  \ket{1_A}=\frac{\ket{0_c}-\ket{1_c}}{\sqrt2},
  \qquad
  \ket{0_B}=\frac{\ket{0_c}+i\ket{1_c}}{\sqrt2},\quad
  \ket{1_B}=\frac{\ket{0_c}-i\ket{1_c}}{\sqrt2},
\]
where the subscript $c$ denotes the computational basis.  Fix the product
stabilizer state
\begin{equation}\label{eq:kitaev-input}
  \ket\Phi=\bigotimes_{u\in A}\ket{0_A}_u\ \bigotimes_{v\in B}\ket{0_B}_v .
\end{equation}
In this local product basis $\ket\Phi$ is the all-zero string.  Write
$\ket{e_w}$ for the basis vector with a single $1$ at vertex $w$, meaning
$\ket{1_A}$ if $w\in A$ and $\ket{1_B}$ if $w\in B$, and identify basis vectors
with elements of $\F_2^{V}$.

\begin{lemma}[Action of the bond terms]\label{lem:bond-action}
With the above conventions, for $e=(u,v)$ with $u\in A$ and $v\in B$,
\begin{equation}\label{eq:bond-action}
  P_e\ket\Phi=
  \begin{cases}
    i\,\ket{e_v},& e\in E_x,\\
    -i\,\ket{e_u},& e\in E_y,\\
    \ket{e_u+e_v},& e\in E_z .
  \end{cases}
\end{equation}
\end{lemma}

\begin{proof}
On $A$ we have $X\ket{0_A}=\ket{0_A}$,
$Y\ket{0_A}=-i\ket{1_A}$ and $Z\ket{0_A}=\ket{1_A}$.  On $B$ we have
$Y\ket{0_B}=\ket{0_B}$, $X\ket{0_B}=i\ket{1_B}$ and
$Z\ket{0_B}=\ket{1_B}$.  Applying these to the two factors of $P_e$ and using
that $P_e$ acts as the identity on all other vertices gives
\eqref{eq:bond-action}.
\end{proof}

Since each $v\in B$ lies on exactly one $x$-bond and each $u\in A$ on exactly
one $y$-bond, the $n$ weight-one strings $\{e_w:w\in V\}$ all occur in
\eqref{eq:bond-action}, and they are distinct from the $n/2$ weight-two strings
coming from $E_z$.  Hence, for $J$ with all $J_e\ne0$, the vector
$H_K(J)\ket\Phi$ has support exactly the following set of bit-string labels in
the local product basis:
\begin{equation}\label{eq:support}
  \cS=\{e_w:w\in V\}\cup\{e_u+e_v:(u,v)\in E_z\},
  \qquad |\cS|=M,
\end{equation}
with all $M$ amplitudes of modulus $|J_e|$.  Thus $\cS$ consists of all
single-vertex excitations, together with the two-vertex excitations supported
on the endpoints of $z$-bonds.

\begin{lemma}[Trivial output stabilizer]\label{lem:kitaev-trivial-stab}
Let $n\ge4$ and let all $J_e\ne0$.  Then the normalized state
\[
  \ket{\Psi_J}=\frac{H_K(J)\ket\Phi}{\|H_K(J)\ket\Phi\|}
\]
is stabilized by no Pauli operator other than the identity; equivalently
$\nu(\ket{\Psi_J})=n$.
\end{lemma}

\begin{proof}
Let $P=\zeta X^{r}Z^{z}$ with $\zeta$ a phase, $r,z\in\F_2^{V}$, and suppose
$P\ket{\Psi_J}=\ket{\Psi_J}$.  For $r\in\F_2^V$ write
$\cS+r:=\{s+r:s\in\cS\}$, with addition in $\F_2^V$.  Since $X^{r}$ permutes the
local product basis by translation and $Z^{z}$ acts diagonally, the support of
$P\ket{\Psi_J}$ is $\cS+r$.  Equality of the two states forces
\begin{equation}\label{eq:support-invariance}
  \cS+r=\cS .
\end{equation}

\emph{Step 1: $r=0$.}  For a bit string $r\in\F_2^V$, let
$\supp(r)=\{u\in V:r_u=1\}$ and let $|r|$ denote its Hamming weight.  Suppose
$r\ne0$ and pick $u\in\supp(r)$.  Then
$e_u\in\cS$, and by \eqref{eq:support-invariance} also $e_u+r\in\cS$.  Every
element of $\cS$ has Hamming weight $1$ or $2$, so
\begin{equation}\label{eq:weight-constraint}
  |e_u+r|=|r|-1\in\{1,2\},\qquad\text{and hence}\qquad 2\le |r|\le3 .
\end{equation}

Choose $w\notin\supp(r)$; such a vertex exists because $n\ge4>3\ge|r|$.  Then
$e_w\in\cS$ and $|e_w+r|=|r|+1\in\{3,4\}$, so $e_w+r\notin\cS$, contradicting
\eqref{eq:support-invariance}.

Hence $r=0$ and $P=\zeta Z^{z}$; Hermiticity of a stabilizing Pauli forces
$\zeta=\pm1$.

\emph{Step 2: $z=0$.}  Write $\ket{\Psi_J}=\sum_{s\in\cS}c_s\ket s$ with all
$c_s\ne0$.  Then $\zeta Z^{z}\ket{\Psi_J}=\ket{\Psi_J}$ means
$\zeta(-1)^{z\cdot s}=1$ for every $s\in\cS$, i.e.\ $z\cdot s$ is the same
element of $\F_2$ for all $s\in\cS$.  Applying this to the weight-one elements
$e_w$, $w\in V$, shows that $z_w$ is independent of $w$, so $z=0$ or
$z=(1,\ldots,1)$.  If $z=(1,\ldots,1)$ then $z\cdot e_w=1$ while
$z\cdot(e_u+e_v)=0$ for every $(u,v)\in E_z$.  But $\cS$ contains all
single-vertex strings $e_w$ and also at least one $z$-bond string
$e_u+e_v$, so $z\cdot s$ cannot be constant on $\cS$.  Hence $z=0$ and
$P=\zeta I$, and $\zeta=1$ because $-I$ stabilizes no nonzero vector.

Thus $\Stab(\ket{\Psi_J})=\{I\}$ and $\nu(\ket{\Psi_J})=n$.
\end{proof}

\begin{theorem}[Exact and constant-error lower bounds]\label{thm:kitaev-size}
Let $n\ge4$.
\begin{enumerate}[label=(\roman*),ref=(\roman*), leftmargin=*]
\item \label{kitaev-part1} Every exact unitary Clifford$+T$ block encoding of a Kitaev Hamiltonian
$H_K(J)$ with all $J_e\ne0$ has $T$-count at least $n$.
\item \label{kitaev-part2} For the uniform instance $H_{\mathrm{unif}}=\sum_{e\in E}P_e$, every
$(\alpha,a,\eps)$ block encoding with $\eps\le1/4$ satisfies $T(U)\ge n$,
independent of $a$ and $\alpha$.
\end{enumerate}
\end{theorem}

\begin{proof}
For part~\ref{kitaev-part1}, let $U$ be an exact block encoding, so $\alpha A_U=H_K(J)$.  Starting from
the stabilizer state $\ket{0^a}\ket\Phi$ and postselecting the block ancillas
onto $\ket{0^a}$ gives the unnormalized output
\[
  (\proj{0^a}\otimes I)U(\ket{0^a}\ket\Phi)
  =
  \ket{0^a}\otimes A_U\ket\Phi
  =
  \frac{1}{\alpha}\ket{0^a}\otimes H_K(J)\ket\Phi .
\]
This vector is nonzero, and after normalization it is
$\ket{0^a}\otimes\ket{\Psi_J}$.  Hence
\[
  \nu(\ket{0^a}\otimes\ket{\Psi_J})
  =
  \nu(\ket{\Psi_J})
  =
  n
\]
by Lemma~\ref{lem:kitaev-trivial-stab}.  Since the input is a stabilizer state
and Pauli postselection does not increase nullity, the $s=T(U)$ non-Clifford
gates must satisfy $s\ge n$.

The proof of part~\ref{kitaev-part2} is given in
Appendix~\ref{app:kitaev-gap}.
\end{proof}

\subsection{A precision lower bound}\label{subsec:kitaev-precision}

The dependence on $\eps$ comes from a different instance of the family.  We
use the following result on postselected one-qubit state preparation, which we
quote in the form in which it is applied.

\begin{fact}[{\cite[Lemma~5.9]{beverland2020lower}}]\label{fact:bchk}
There are absolute constants such that the following holds.  Let $0<\delta<1/8$
and let $p$ satisfy $\delta<p<3\delta$.  Any stabilizer protocol that, from
stabilizer resources together with $k$ copies of the magic state $\ket T$,
prepares a one-qubit state within trace distance $\delta$ of $\proj{\psi_p}$
with $\ket{\psi_p}=\sqrt p\ket0+\sqrt{1-p}\ket1$, must satisfy
\[
  k\ \ge\ \tfrac13\log_2(1/\delta)-\tfrac23 .
\]
\end{fact}

\begin{theorem}[Precision lower bound]\label{thm:kitaev-precision}
Let $n\ge4$ and $0<\eps<1/32$.  There is $H_\eps\in\cK_n$ with $\|H_\eps\|=1$
such that every $(\alpha,a,\eps)$ unitary Clifford$+T$ block encoding $U$ of
$H_\eps$ obeys
\begin{equation}\label{eq:kitaev-precision-bound}
  T(U)\ \ge\ \tfrac13\log_2\tfrac{1}{4\eps}-\tfrac23 .
\end{equation}
\end{theorem}

\begin{proof}[Proof idea]
Choose two anticommuting bond operators $P$ and $Q$, set
$p=8\eps$, and define
\[
    H_\eps=\sqrt p\,P+\sqrt{1-p}\,Q.
\]
Then $H_\eps^2=I$. For a suitable stabilizer input, applying
$H_\eps$ produces, up to Clifford operations and stabilizer tensor
factors, the one-qubit state
$\ket{\psi_p}=\sqrt p\ket0+\sqrt{1-p}\ket1$.

An $\eps$-approximate block encoding therefore gives a stabilizer
protocol that prepares this state within trace-distance error
$\delta=4\eps$. Since $\delta<p<3\delta$, Fact~\ref{fact:bchk}
gives the claimed lower bound. The complete reduction and error
analysis are given in Appendix~\ref{app:kitaev-precision}.
\end{proof}

\subsection{An explicit block encoding}\label{subsec:kitaev-upper}

We now give a construction matching the lower bounds above.  The point is that
$H_K(J)$ is already a linear combination of $M=\tfrac32n$ two-qubit Pauli
unitaries, so the standard LCU block-encoding form can be used directly
\cite{childs2012lcu}.

Fix $J$ and set $\lambda=\sum_{e\in E}|J_e|$.  If $\lambda\le\eps$ then the
zero block, implemented by $U=X\otimes I$ on one ancilla, already gives an $\eps$
approximation with no $T$ gates, so assume $\lambda>\eps$.  Let
$\ell=\lceil\log_2M\rceil$ and define on an $\ell$-qubit address register
\begin{equation}\label{eq:prepare-state}
  \ket G=\sum_{e\in E}\sqrt{\frac{|J_e|}{\lambda}}\,\ket e .
\end{equation}
With $\eta_e=\operatorname{sgn}(J_e)$ for $J_e\ne0$ and $\eta_e=1$ otherwise,
define
\begin{equation}\label{eq:kitaev-select}
  \mathrm{SEL}=\sum_{e\in E}\proj e\otimes\eta_eP_e
   +\sum_{e\notin E}\proj e\otimes I ,
\end{equation}
the second sum ranging over unused binary addresses.  If
$\mathrm{PREP}\ket{0^{\ell}}=\ket G$ then
\begin{equation}\label{eq:lcu-identity}
  (\bra{0^{\ell}}\otimes I)\,
  (\mathrm{PREP}^{\dagger}\otimes I)\,\mathrm{SEL}\,(\mathrm{PREP}\otimes I)\,
  (\ket{0^{\ell}}\otimes I)
  =\sum_{e\in E}\frac{|J_e|}{\lambda}\eta_eP_e=\frac{H_K(J)}{\lambda},
\end{equation}
which is the usual LCU block-encoding identity.

Using the state-preparation construction of
Ref.~\cite{gosset2026stateprep}, the implementation and error analysis
in Appendix~\ref{app:kitaev-blockencoding} show that the LCU circuit
can be chosen as a $(\lambda,a,\eps)$ block encoding with $a=O(n+\log(1/\eps))$ and
\begin{equation}
\label{eq:kitaev-upper}
\begin{aligned}
    T(U)
    &=
    O\!\left(
        n+\sqrt{n\log(n/\eps)}+\log(n/\eps)
    \right) \\
    &=
    O\!\left(n+\log(1/\eps)\right).
\end{aligned}
\end{equation}
\begin{theorem}[Optimal Kitaev worst-case scaling]\label{thm:kitaev-optimal}
For $n\ge4$ and $0<\eps<1/32$,
\begin{equation}\label{eq:kitaev-theta}
  T^{\BE}_{\eps}(\cK_n)=\Theta\!\left(n+\log(1/\eps)\right).
\end{equation}
\end{theorem}

\begin{proof}
The upper bound \eqref{eq:kitaev-upper} holds for every $H_K(J)\in\cK_n$, hence
bounds the supremum in \eqref{eq:worstcase}.  For the lower bound,
part~\ref{kitaev-part2} of Theorem~\ref{thm:kitaev-size} exhibits, in the uniform instance
$H_{\mathrm{unif}}$, a Hamiltonian requiring at least $n$ $T$ gates (note
$\eps<1/32<1/4$), and Theorem~\ref{thm:kitaev-precision} exhibits a
two-bond instance \eqref{eq:precision-H} requiring at least
$\tfrac13\log_2(1/4\eps)-\tfrac23$.  Since \eqref{eq:worstcase} takes a
supremum over the family,
\[
  T^{\BE}_{\eps}(\cK_n)\ \ge\ \max\bigl\{n,\ \tfrac13\log_2\tfrac1{4\eps}-\tfrac23\bigr\}
  \ \ge\ \tfrac12\Bigl(n+\tfrac13\log_2\tfrac1{4\eps}-\tfrac23\Bigr)
  =\Omega\!\left(n+\log(1/\eps)\right),
\]
using $\max\{x,y\}\ge\tfrac12(x+y)$.
\end{proof}

\begin{corollary}[Fixed normalization]
\label{cor:kitaev-fixed-normalization}
Recall that
\[
  M=|E|=\frac{3n}{2}.
\]
Every Hamiltonian in $\cK_n$ admits a block encoding with the common
normalization $M$. For this normalization,
\[
  T^{\BE}_{\eps,M}(\cK_n)
  =
  \Theta\!\left(n+\log(1/\eps)\right).
\]
The matching construction uses $O(n+\log(1/\eps))$ ancillas.
\end{corollary}

\begin{proof}
The lower bound follows from the unrestricted-normalization result.
For the upper bound, we pad the LCU decomposition with two cancelling
identity terms so that its total weight is $M$. This does not change
the asymptotic $T$-count or ancilla count. The details are given in
Appendix~\ref{app:common-normalization}.
\end{proof}

\section{Application: T-Count of fixed QSVT Hamiltonian Simulation}
\label{sec:qsp-application}
A block encoding is often used as an input to another quantum
algorithm. Here we apply our results to the standard QSVT construction
for Hamiltonian simulation
\cite[Theorem~58]{gilyen2019qsvt}. The two preceding corollaries~\ref{cor:secondq-fixed-normalization}
and~\ref{cor:kitaev-fixed-normalization}
give a common normalization for each Hamiltonian family. Thus, for
fixed simulation time and error, the same QSVT construction can be used
throughout the family.

\subsection{The QSVT circuit}

Let $U$ be an $(\alpha,a,\eps_{\mathrm{BE}})$ block encoding of $H$,
and write
\[
  \widetilde H:=\alpha A_U.
\]
By definition,
\[
  \|H-\widetilde H\|\le\eps_{\mathrm{BE}}.
\]
For the two constructions considered here, $\widetilde H$ can be
chosen Hermitian; see
Appendix~\ref{app:qsp-hermitian-approximation}.

Define
\[
  \Pi_a:=|0^a\rangle\langle0^a|\otimes I.
\]
The QSVT circuit alternates calls to $U$ and $U^\dagger$ with phase
operators
\[
  \Pi_\phi:=e^{i\phi(2\Pi_a-I)}.
\]
Each $\Pi_\phi$ can be implemented using two
projector-controlled NOT gates and one single-qubit $Z$ rotation
\cite{gilyen2019qsvt,martyn2021grand}.

Following Ref.~\cite[Lemma~57 and Theorem~58]{gilyen2019qsvt},
we use even and odd polynomial approximations to
$\cos(\alpha t x)$ and $\sin(\alpha t x)$, respectively,
with degrees chosen to achieve the query bound in
Eq.~\eqref{eq:qsp-query-count}.
Following the proof of Ref.~\cite[Theorem~58]{gilyen2019qsvt},
we coherently combine these approximations with relative phase $-i$
to obtain a unitary $V=\mathcal U_{\boldsymbol{\Phi}}$ satisfying
\[
  (\langle0^{a+2}|\otimes I)V(|0^{a+2}\rangle\otimes I)
  \approx \frac{e^{-it\widetilde H}}{2}.
\]
Robust oblivious amplitude amplification (OAA)
then yields an $\eps_{\mathrm{QSVT}}$-approximate block encoding
of $e^{-it\widetilde H}$. 
The circuit before OAA is shown in
Fig.~\ref{fig:qsvt-simulation}.

\begin{figure}[t]
  \centering

  \begin{subfigure}[c]{0.25\textwidth}
    \centering
    \includegraphics[width=\linewidth]{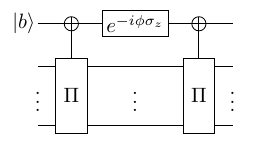}
    \caption{$\ket{b}\bra{b}\otimes e^{(-1)^{b}i\phi(2\Pi_a-I)}$}
    \label{fig:qsvt-phase}
  \end{subfigure}
  \hfill
  \begin{subfigure}[c]{0.70\textwidth}
    \centering
    \includegraphics[width=\linewidth]
      {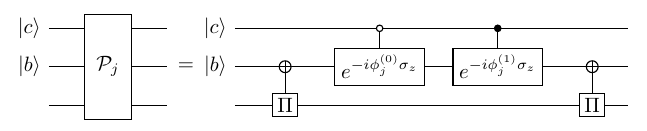}
    \caption{$\ket{cb}\bra{cb}\otimes e^{(-1)^{b}i\phi^{c}_j(2\Pi_a-I)}$}
    \label{fig:qsvt-selected-phase}
  \end{subfigure}

  \vspace{0.8em}

  \begin{subfigure}[c]{0.98\textwidth}
    \centering
    \includegraphics[width=\linewidth]
      {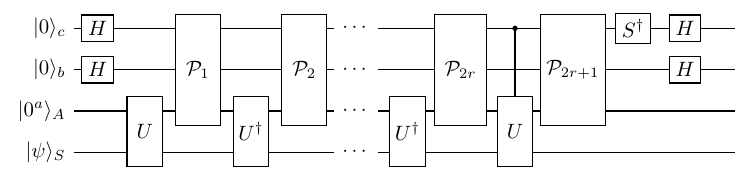}
    \caption{$\displaystyle \mathcal U_{\boldsymbol{\Phi}}
  = (H_c\otimes H_b)S_c^\dagger\,
  \mathcal P_{2r+1}\,C_c(U)\, \overleftarrow{\prod}_{j=1}^{r}
  \left(\mathcal P_{2j}U^\dagger \mathcal P_{2j-1}U\right)
  (H_c\otimes H_b)
  $}
    \label{fig:qsvt-sequence}
  \end{subfigure}

  \caption{
    The QSVT circuit $\mathcal U_{\boldsymbol{\Phi}}$ approximately block-encodes $e^{-it\widetilde H}/2$ before OAA.
    (a) Decomposition of a projector-phase operator into two
    projector-controlled NOT gates and a single-qubit $Z$ rotation, adapted from Ref.~\cite[Fig.~1(b)]{gilyen2019qsvt}.
    (b) The selected phase operator $\mathcal P_j$, adapted from
    Ref.~\cite[Fig.~1(c)]{gilyen2019qsvt}. The auxiliary qubit $c$ here selects the phase sequence for the cosine
    or sine polynomial.
    (c) The complete QSVT alternating sequence. The register $A$ contains the $a$ block-encoding ancilla qubits,
and $S$ is the system register. The final controlled
    query accounts for the difference between the even cosine
    polynomial and the odd sine polynomial, while $S^\dagger$
    supplies the factor $-i$ on the sine branch.
    Here, $r$ is the number of query pairs, giving $2r$ uncontrolled queries before OAA.
  }
  \label{fig:qsvt-simulation}
\end{figure}

Throughout this section, $N_{\rm calls}$ counts all uncontrolled
uses of $U$ and $U^\dagger$ in the complete simulation circuit,
including OAA. The construction also uses three controlled
queries to $U$ or $U^\dagger$ \cite[Theorem~58]{gilyen2019qsvt},
which we count separately from $N_{\rm calls}$.

\subsection{Modular T-count}

We count each occurrence of $U$ or $U^\dagger$ as a separate circuit
module without recompiling several calls into a new circuit. And we call
the resulting cost the modular $T$-count.

For the two Hamiltonian families, define
\[
  B_2(n,\eps_{\mathrm{BE}})
  :=
  n^2\sqrt{\log\!\left(\frac{n^4}{\eps_{\mathrm{BE}}}\right)},
  \qquad
  B_{\rm K}(n,\eps_{\mathrm{BE}})
  :=
  n+\log\!\left(\frac{1}{\eps_{\mathrm{BE}}}\right).
\]
These are the optimal costs of one block-encoding module under the conditions stated above at the common
normalizations given in
Corollaries~\ref{cor:secondq-fixed-normalization}
and~\ref{cor:kitaev-fixed-normalization}.
And let $C_F$ denote the $T$-count of one controlled-$U$ or controlled-$U^{\dagger}$, where $U$ is the corresponding block encoding.

\begin{theorem}[Cost of a fixed QSVT circuit]
\label{thm:qsp-modular-cost}
Let $F\in\{2,{\rm K}\}$, and let $B_F$ denote the corresponding
block-encoding cost above. Let $C_F$ denote the $T$-count of one controlled-$U$ or controlled-$U^\dagger$ query. Suppose the QSVT circuit makes
$N_{\rm calls}$ calls to $U$ and $U^\dagger$. Then its worst-case
modular $T$-count satisfies
\[
  T^{\mathrm{mod}}_{\mathrm{QSVT},F}
  =
  \Omega\!\left(N_{\mathrm{calls}}B_F\right),
  \qquad
  T^{\mathrm{mod}}_{\mathrm{QSVT},F}
  =
  O\!\left(
    N_{\mathrm{calls}}
    \left[
      B_F+
      \log\!\left(
        \frac{N_{\mathrm{calls}}}{\eps_{\mathrm{syn}}}
      \right)
    \right]
    + C_F
  \right).
\]
where $\eps_{\mathrm{syn}}$ is the total error from synthesizing the
QSVT phase rotations.

In particular, if
\[
  \log\!\left(
    \frac{N_{\rm calls}}{\eps_{\mathrm{syn}}}
  \right)
  =
  O(B_F),
  \qquad\text{and}\qquad
C_F=O(N_{\rm calls}B_F),
\]
then
\[
  T^{\rm mod}_{\rm QSVT,F}
  =
  \Theta\!\left(N_{\rm calls}B_F\right).
\]
\end{theorem}

\begin{proof}
The $N_{\rm calls}$ block-encoding modules give the lower bound. The
same modules contribute $O(N_{\rm calls}B_F)$ $T$ gates to the upper
bound.

The projector-controlled NOT gates in the phase operators cost $O(a)$ $T$ gates,
where $a=O(B_F)$ for both constructions. The complete circuit contains $O(N_{\rm calls})$ single-qubit Z rotations, including those used to select between the cosine and sine branches. Synthesizing these rotations to total error
$\eps_{\mathrm{syn}}$ costs
\[
  O\!\left(
    N_{\rm calls}
    \log\!\frac{N_{\rm calls}}{\eps_{\mathrm{syn}}}
  \right).
\]
The three controlled block-encoding queries contribute $O(C_F)$
additional $T$ gates. The remaining gates in OAA can be implemented using
$O(a+1)$ $T$ gates, which is absorbed by $O(N_{\rm calls}B_F)$ since
$a+1=O(B_F)$; see Appendix~\ref{app:qsp-oaa-cost}.
Combining these costs proves the result.
\end{proof}

\subsection{Hamiltonian simulation}

We now apply Theorem~\ref{thm:qsp-modular-cost} to the standard QSVT
algorithm for Hamiltonian simulation
\cite{gilyen2019qsvt,martyn2021grand}. The regime $\alpha|t|\le\eps_{\mathrm{QSVT}}$ is trivial, since the
identity already approximates the evolution within the target error.
We therefore consider $\alpha|t|>\eps_{\mathrm{QSVT}}$.

Since $U$ exactly block encodes $\widetilde H$ with normalization
$\alpha$, Ref.~\cite[Corollary~60]{gilyen2019qsvt} applies. The construction also requires three uses of
controlled-$U$ or its inverse
\cite[Theorem~58]{gilyen2019qsvt}. For
$0<\eps_{\mathrm{QSVT}}<1/2$, the query-optimal construction makes
\begin{equation}\label{eq:qsp-query-count}
  N_{\rm calls}
  =
  \Theta\!\left(
    \alpha|t|
    +
    \frac{\log(1/\eps_{\mathrm{QSVT}})}
    {\log\!\left(
      e+
      \frac{\log(1/\eps_{\mathrm{QSVT}})}{\alpha|t|}
    \right)}
  \right)
\end{equation}
queries to $U$ and $U^\dagger$.

The target evolution is $e^{-itH}$ rather than
$e^{-it\widetilde H}$. Since $H$ and $\widetilde H$ are Hermitian,
\[
  \left\|
    e^{-itH}-e^{-it\widetilde H}
  \right\|
  \le
  |t|\eps_{\mathrm{BE}}.
\]
A proof of this bound is given in Appendix~\ref{app:qsp-cost}.

The total simulation error is therefore at most
\[
  |t|\eps_{\mathrm{BE}}
  +
  \eps_{\mathrm{QSVT}}
  +
  \eps_{\mathrm{syn}}.
\]
To achieve error $\eps_{\mathrm{sim}}$, it is enough to choose
\[
  |t|\eps_{\mathrm{BE}}
  \le\frac{\eps_{\mathrm{sim}}}{3},
  \qquad
  \eps_{\mathrm{QSVT}}
  =\frac{\eps_{\mathrm{sim}}}{3},
  \qquad
  \eps_{\mathrm{syn}}
  =\frac{\eps_{\mathrm{sim}}}{3}.
\]

Because $\alpha$ is fixed within each Hamiltonian family, the same
query bound applies to every Hamiltonian in that family.

\subsection{Second-quantized Hamiltonians}

We first consider $\cH_n$, whose common normalization is
$\overline\alpha_{2,n}=\Theta(n^4)$ by
Corollary~\ref{cor:secondq-fixed-normalization}.

\begin{corollary}
\label{cor:qsp-secondq}
Let $t > 0$ and $\eps_{\mathrm{sim}}\in(0,1)$.
Choose $\eps_{\mathrm{BE}}\in(0,1]$ such that
$|t|\eps_{\mathrm{BE}}\le\eps_{\mathrm{sim}}/3$, and set
$\eps_{\mathrm{QSVT}}=\eps_{\mathrm{syn}}=\eps_{\mathrm{sim}}/3$.
Assume $\log(n^4/\eps_{\mathrm{BE}})=O(n^4)$.
Let $N_{\rm calls}^{(2)}$ be given by
Eq.~\eqref{eq:qsp-query-count} with
$\alpha=\overline\alpha_{2,n}$.

The QSVT construction implements an
$\eps_{\mathrm{sim}}$-approximate block encoding of $e^{-itH}$
for every $H\in\cH_n$. If
\[
  \log\!\left(
    \frac{3N_{\rm calls}^{(2)}}{\eps_{\mathrm{sim}}}
  \right)=O(B_2),
  \qquad
  B_2=O(N_{\rm calls}^{(2)}),
\]
where $B_2=B_2(n,\eps_{\mathrm{BE}})$, then its worst-case
modular $T$-count satisfies
\[
  T^{\mathrm{mod}}_{\mathrm{QSVT},2}
  =
  \Theta\!\left(
    N_{\rm calls}^{(2)}
    n^2\sqrt{\log\!\left(
      \frac{n^4}{\eps_{\mathrm{BE}}}
    \right)}
  \right).
\]
\end{corollary}

\begin{proof}
Theorem~\ref{thm:secondq-lower} and the padded construction in
Appendix~\ref{app:secondq-common-normalization} give the
matching block-encoding bounds in the stated regime.
The appendix also gives
$C_2=O(B_2^2)=O(N_{\rm calls}^{(2)}B_2)$.
Together with the phase-synthesis assumption above,
Theorem~\ref{thm:qsp-modular-cost} yields the result.
\end{proof}

\subsection{Kitaev Hamiltonians}

We next consider $\cK_n$. By
Corollary~\ref{cor:kitaev-fixed-normalization}, its common
normalization is $M=3n/2$.

\begin{corollary}
\label{cor:qsp-kitaev}
Let $t>0$ and $\eps_{\mathrm{sim}}\in(0,1)$.
Choose $\eps_{\mathrm{BE}}\in(0,1/32)$ such that
$|t|\eps_{\mathrm{BE}}\le\eps_{\mathrm{sim}}/3$, and set
$\eps_{\mathrm{QSVT}}=\eps_{\mathrm{syn}}=\eps_{\mathrm{sim}}/3$.
Let $N_{\rm calls}^{({\rm K})}$ be given by
Eq.~\eqref{eq:qsp-query-count} with $\alpha=M$.

The QSVT construction implements an
$\eps_{\mathrm{sim}}$-approximate block encoding of $e^{-itH}$
for every $H\in\cK_n$. If
\[
  \log\!\left(
    \frac{3N_{\rm calls}^{({\rm K})}}{\eps_{\mathrm{sim}}}
  \right)
  =O(B_{\rm K}(n,\eps_{\mathrm{BE}})),
\]
then its worst-case modular $T$-count satisfies
\[
  T^{\mathrm{mod}}_{\mathrm{QSVT},{\rm K}}
  =
  \Theta\!\left(
    N_{\rm calls}^{({\rm K})}
    \left[n+\log\!\left(
      \frac{1}{\eps_{\mathrm{BE}}}
    \right)\right]
  \right).
\]
\end{corollary}

\begin{proof}
Since $C_{\rm K}=O(B_{\rm K})$ and
$N_{\rm calls}^{({\rm K})}\ge1$, the controlled-query cost
is absorbed by $N_{\rm calls}^{({\rm K})}B_{\rm K}$.
The conclusion follows from
Corollary~\ref{cor:kitaev-fixed-normalization} and
Theorem~\ref{thm:qsp-modular-cost}.
\end{proof}

\section{Conclusion and outlook}\label{sec:discussion}

We determine the optimal $T$-count for block encodings of one
fermionic Hamiltonian family at fixed precision and one spin Hamiltonian family. The two
lower bounds require different methods. For the general
second-quantized family, the parameter dimension
$d=\Theta(n^4)$ makes volume covering effective once ancilla
compression bounds the total number of qubits. The Kitaev family has
only $\Theta(n)$ independent couplings, so the same counting argument
is too weak. Instead, stabilizer nullity gives the linear lower bound,
while a two-bond reduction gives the logarithmic precision bound.

The Kitaev lower-bound strategy is not tied to the honeycomb geometry.
The linear lower-bound argument extends to other Pauli spin families
whenever some allowed Hamiltonian maps a stabilizer input to a state
with stabilizer nullity $\Omega(n)$. And the approximation must still force the postselected output to have
stabilizer nullity $\Omega(n)$. The logarithmic precision argument
applies when the family admits a similar reduction to one-qubit state
preparation. These conditions are model dependent.

Our QSVT application shows how the optimal cost of one block-encoding
query contributes to the $T$-count of a fixed Hamiltonian simulation
circuit before any additional compilation. It does not give a lower bound for arbitrary simulation
architectures or for a global recompilation of the full circuit. This
also differs from the lower bounds of Zlokapa \emph{et al.} on
two-qubit gate complexity \cite{zlokapa2026optimal}, whose model allows
arbitrary single-qubit gates at no cost. Their precision dependence
therefore does not directly imply a $T$-count lower bound.

Several questions remain open. Our lower bounds do not cover adaptive
circuits with mid-circuit measurements, and the fermionic result does
not imply the same scaling for molecular Hamiltonians. It would be
interesting to obtain comparable bounds for physically constrained
fermionic families and for broader classes of spin models. More
generally, our results show that term count alone does not determine
block-encoding cost: the dimension and algebraic structure of the
Hamiltonian family are also essential.

\section*{Author contributions}
JM developed and proved the lower bounds for the second-quantized Hamiltonians, as well as the lower and upper bounds for the Kitaev Hamiltonian, contributed to the analysis of the QSVT application, and participated in writing the manuscript. KJ and YL contributed to discussions on the lower-bound arguments, writing and reviewing the manuscript. KJ generated the T-count figure for realistic molecular systems. YL conceived the idea and supervised the project. 

ChatGPT was used as an auxiliary tool in developing arguments related to the ancilla-compression theorem and the $\Omega(n)$ complexity bound for the Kitaev Hamiltonian, as well as for assistance with writing and language editing.

\bibliographystyle{quantum}
\bibliography{main}

\appendix

\section{Symplectic preliminaries}\label{app:symplectic}

This appendix collects the facts about the binary symplectic representation
used in Appendix~\ref{app:pauli-compression}.  Everything is finite-dimensional
linear algebra over $\F_2$; see, e.g., Refs.~\cite{dehaene2003clifford,gottesman1997thesis}.

Ignoring the overall phases $\{\pm1,\pm i\}$, every $a$-qubit Pauli operator can
be written uniquely as $X^{x}Z^{z}$ with $x,z\in\F_2^{a}$, and we set
$\nu(X^{x}Z^{z})=(x,z)\in V:=\F_2^{2a}$ as in \eqref{eq:nu}.  Two Pauli
operators $P,P'$ satisfy $PP'=(-1)^{\omega(\nu(P),\nu(P'))}P'P$, where
\begin{equation}\label{eq:symp-form}
  \omega\bigl((x,z),(x',z')\bigr)=x\cdot z'+z\cdot x'\pmod 2 .
\end{equation}
The symplectic form $\omega$ is bilinear, alternating and nondegenerate.  For a subspace
$W\subseteq V$ its symplectic complement is
$W^{\perp}=\{v\in V:\omega(v,w)=0\ \forall w\in W\}$.

\begin{lemma}\label{lem:dim-perp}
For every subspace $W\subseteq V$, $\dim W+\dim W^{\perp}=2a$.
\end{lemma}

\begin{proof}
Consider $\Phi_W:V\to W^{*}$, $\Phi_W(v)=\omega(v,\cdot)|_W$, whose kernel is
$W^{\perp}$ by definition.  Nondegeneracy of $\omega$ means that
$v\mapsto\omega(v,\cdot)$ is an isomorphism $V\to V^{*}$, and the restriction
map $V^{*}\to W^{*}$ is surjective; hence $\Phi_W$ is surjective.  Therefore
\[
  2a=\dim V=\dim\ker\Phi_W+\dim W^{*}
  =\dim W^{\perp}+\dim W,
\]
as claimed.
\end{proof}

A family $h_1,g_1,\ldots,h_a,g_a$ is a \emph{symplectic basis} of $V$ if it is a
basis and $\omega(h_i,g_j)=\delta_{ij}$, $\omega(h_i,h_j)=\omega(g_i,g_j)=0$.
The standard symplectic basis is $e_i=(\mathbf e_i,0)$, $f_i=(0,\mathbf e_i)$,
corresponding to $X_i$ and $Z_i$.  A linear map $S:V\to V$ is
\emph{symplectic} if $\omega(Sv,Sv')=\omega(v,v')$ for all $v,v'$; equivalently,
$S$ maps some (hence any) symplectic basis to a symplectic basis.  The group of
such maps is $\Sp(2a,\F_2)$.

The subspace $L_Z$ of \eqref{eq:LZ} is the image under $\nu$ of the stabilizer
group of $\ket{0^{a}}$, and any two vectors in $L_Z$ are symplectically
orthogonal.

\begin{fact}[Clifford--symplectic correspondence]\label{fact:cliff-symp}
Conjugation by a Clifford operator $C\in\Cl_a$ induces a symplectic map
$S_C\in\Sp(2a,\F_2)$ through $\nu(CPC^{\dagger})=S_C\,\nu(P)$, and the induced
homomorphism $\Cl_a\to\Sp(2a,\F_2)$ is surjective with kernel generated by the
Pauli group and global phases \cite{dehaene2003clifford}.
\end{fact}

\section{Pauli-support compression}\label{app:pauli-compression}

Throughout this appendix $V=\F_2^{2a}$ with the notation of
Appendix~\ref{app:symplectic}.  We write $\cW$ for the span of the ancilla parts
of the Pauli operators appearing in the rotations $R(P_j)$ of the canonical
form \eqref{eq:pauli-normal-form}, reserving $W$ for the Clifford appearing in
Appendix~\ref{app:clifford-corner}.

\subsection{The symplectic compression lemma}

\begin{lemma}[Symplectic support compression]\label{lem:symplectic-compression}
Let $\cW\subseteq V$ be a subspace of dimension $r \le a$.  There exists
$S\in\Sp(2a,\F_2)$ such that
\begin{equation}\label{eq:symp-goal}
  S(L_Z)=L_Z
  \qquad\text{and}\qquad
  S(\cW)\subseteq\Span\{e_1,f_1,\ldots,e_r,f_r\}.
\end{equation}
\end{lemma}

\begin{proof}
The idea is to construct a symplectic basis $h_1,g_1,\ldots,h_a,g_a$ of $V$
such that $g_1,\ldots,g_a$ is a basis of $L_Z$ and such that the last $a-r$
symplectic pairs $(h_i,g_i)$, $i>r$, lie in $\cW^{\perp}$.  The map sending
$h_i\mapsto e_i$, $g_i\mapsto f_i$ then has both required properties.

\emph{Step 1: dimension count.}
Put
\[
  K_0=L_Z\cap\cW^{\perp},\qquad R_0=\cW\cap\cW^{\perp},
\]
and let $\pi_X:V\to\F_2^{a}$, $\pi_X(x,z)=x$, be the projection onto the $X$
part, with $d_X:=\dim\pi_X(\cW)$.  Since $\ker(\pi_X|_{\cW})=L_Z\cap\cW$,
\begin{equation}\label{eq:LZ-cap-W}
  \dim(L_Z\cap\cW)=r-d_X .
\end{equation}
A vector $(0,z)\in L_Z$ lies in $\cW^{\perp}$ if and only if
$\omega((0,z),(x,z'))=z\cdot x=0$ for all $(x,z')\in\cW$, that is, if and only
if $z\in\pi_X(\cW)^{\perp_*}$ with respect to the ordinary dot product on
$\F_2^{a}$.  Hence
\begin{equation}\label{eq:dim-K}
  \dim K_0=a-d_X .
\end{equation}
Because $K_0\cap R_0\subseteq L_Z\cap\cW$, equations \eqref{eq:LZ-cap-W} and
\eqref{eq:dim-K} give
\[
  \dim K_0-\dim(K_0\cap R_0)\ \ge\ (a-d_X)-(r-d_X)=a-r .
\]
We may therefore choose an $(a-r)$-dimensional subspace $G\subseteq K_0$ with
$G\cap R_0=\{0\}$.  Since $G\subseteq K_0\subseteq\cW^{\perp}$, any vector of
$G\cap\cW$ would lie in $\cW\cap\cW^{\perp}=R_0$, so in fact
\begin{equation}\label{eq:G-cap-W}
  G\cap\cW=\{0\}.
\end{equation}
Fix a basis $g_{r+1},\ldots,g_a$ of $G$.  Since $G\subseteq L_Z$, these vectors
are mutually symplectically orthogonal.

\emph{Step 2: dual vectors in $\cW^{\perp}$.}
Define $\Psi:\cW^{\perp}\to G^{*}$ by $\Psi(h)=\omega(h,\cdot)|_G$.  Its kernel
is $G^{\perp}\cap\cW^{\perp}=(G\cup\cW)^{\perp}$.  Since
$\Span(G\cup\cW)=G\oplus\cW$ with $G\cap\cW=\{0\}$ by \eqref{eq:G-cap-W},
Lemma~\ref{lem:dim-perp} gives
$\dim\ker\Psi=2a-\dim(G\oplus\cW)=2a-a=a$.  Since
$\dim\cW^{\perp}=2a-r$,
\[
  \dim\im\Psi=\dim\cW^\perp-\dim\ker\Psi=a-r=\dim G^{*},
\]
so $\Psi$ is surjective.  Choose $\widetilde h_{r+1},\ldots,\widetilde h_a\in\cW^{\perp}$
with
\[
  \omega(\widetilde h_i,g_j)=\delta_{ij},\qquad r+1\le i,j\le a .
\]
These need not be mutually orthogonal.  Correct them in descending order by
\begin{equation}\label{eq:tail-orthogonalization}
  h_i=
  \begin{cases}
    \displaystyle
    \widetilde h_i+\sum_{j=i+1}^{a}\omega(\widetilde h_i,\widetilde h_j)\,g_j,
    & r+1\le i<a,\\[0.8em]
    \widetilde h_a, & i=a,
  \end{cases}
\end{equation}
Every $g_j$ is symplectically orthogonal to every $g_l$, so the duality relations are preserved:
$\omega(h_i,g_j)=\delta_{ij}$.  Moreover $h_i\in\cW^{\perp}$ because both
$\widetilde h_i$ and the $g_j$ lie there.  For $i<k$, using
$\omega(h_i,g_l)=\delta_{il}=0$ for $l \geq i$ except $l=i$, and
$\omega(g_j,\widetilde h_k)=\delta_{jk}$,
\[
  \omega(h_i,h_k)=\omega(h_i,\widetilde h_k)
  =\omega(\widetilde h_i,\widetilde h_k)
   +\sum_{j>i}\omega(\widetilde h_i,\widetilde h_j)\,\omega(g_j,\widetilde h_k)
  =\omega(\widetilde h_i,\widetilde h_k)+\omega(\widetilde h_i,\widetilde h_k)=0
\]
over $\F_2$.  Hence $h_{r+1},g_{r+1},\ldots,h_a,g_a$ form $a-r$ symplectic
pairs, all of whose members lie in $\cW^{\perp}$.

\emph{Step 3: completing the basis of $L_Z$.}
Extend $g_{r+1},\ldots,g_a$ to a basis
$\widetilde g_1,\ldots,\widetilde g_r,g_{r+1},\ldots,g_a$ of $L_Z$, and for
$1\le i\le r$ set
\begin{equation}\label{eq:head-g-correction}
  g_i=\widetilde g_i+\sum_{k=r+1}^{a}\omega(\widetilde g_i,h_k)\,g_k .
\end{equation}
This is again a basis of $L_Z$, and since $\omega(g_k,h_l)=\delta_{kl}$ for
$k,l>r$ we obtain
\[
  \omega(g_i,h_k)=0,\qquad 1\le i\le r<k\le a .
\]

\emph{Step 4: completing the symplectic basis.}
Let $U_0=\Span\{h_{r+1},g_{r+1},\ldots,h_a,g_a\}$.  The restriction of $\omega$
to $U_0$ is nondegenerate because $U_0$ is spanned by symplectic pairs, so
$U_0^{\perp}$ is a symplectic space of dimension $2a-2(a-r)=2r$, and by Step 3
it contains $g_1,\ldots,g_r$.  Repeating the dual-basis construction and the
orthogonalization \eqref{eq:tail-orthogonalization} inside $U_0^{\perp}$
produces $h_1,\ldots,h_r\in U_0^{\perp}$ with
\[
  \omega(h_i,g_j)=\delta_{ij},\qquad\omega(h_i,h_j)=0,\qquad 1\le i,j\le r .
\]
Then $h_1,g_1,\ldots,h_a,g_a$ is a symplectic basis of $V$.

\emph{Step 5: the transformation.}
Define $S$ by $S(h_i)=e_i$ and $S(g_i)=f_i$ for all $i$.  As $S$ maps a
symplectic basis to a symplectic basis it is symplectic, and since
$L_Z=\Span\{g_1,\ldots,g_a\}$ we get $S(L_Z)=\Span\{f_1,\ldots,f_a\}=L_Z$.
Finally, let $w\in\cW$.  For $i>r$ both $h_i$ and $g_i$ lie in $\cW^{\perp}$, so
\[
  \omega(Sw,e_i)=\omega(Sw,Sh_i)=\omega(w,h_i)=0,
  \qquad
  \omega(Sw,f_i)=\omega(w,g_i)=0 .
\]
These say that both coordinates of $Sw$ in the $i$th symplectic pair vanish for
every $i>r$, that is, $Sw\in\Span\{e_1,f_1,\ldots,e_r,f_r\}$.
\end{proof}

\subsection{Lifting to a Clifford, and preservation of the block}

\begin{proposition}[Pauli-support compression]\label{prop:pauli-compression}
Assume $a>s$.
Let $U$ be written as in \eqref{eq:pauli-normal-form} on $a$ ancillas and $n$
system qubits, let $P_j=B_j\otimes Q_j$ be the ancilla--system splitting of its
Pauli generators, and put
\[
  r=\dim_{\F_2}\Span\{\nu(B_1),\ldots,\nu(B_s)\}\ \le\ s .
\]
Then there is a Clifford $D_A$ on the ancilla register with
\begin{enumerate}[label=(\roman*)]
\item $D_A\ket{0^{a}}=\ket{0^{a}}$, and
\item $D_AB_jD_A^{\dagger}$ acts trivially on the last $a-r$ ancillas for
every $j$.
\end{enumerate}
Moreover, with $D=D_A\otimes I$, the circuit $DUD^{\dagger}$ has the same
encoded block as $U$, hence is again an $(\alpha,a,\eps)$ block encoding of the
same operator.
\end{proposition}

\begin{proof}
Apply Lemma~\ref{lem:symplectic-compression} to
$\cW=\Span\{\nu(B_1),\ldots,\nu(B_s)\}$, obtaining $S\in\Sp(2a,\F_2)$ with
\eqref{eq:symp-goal}.  By Fact~\ref{fact:cliff-symp} there is a Clifford
$\widetilde D_A$ inducing $S$.

Since $S(L_Z)=L_Z$ and $L_Z$ is the image of the stabilizer group of
$\ket{0^{a}}$, the Clifford $\widetilde D_A$ maps that stabilizer group to
itself up to signs of generators; equivalently it normalizes the group of
$Z$-type Paulis.  Thus $\widetilde D_A\ket{0^a}$ is stabilized by a signed
maximal $Z$-type stabilizer group.  Equivalently, for each $i$ there is
$\beta_i\in\F_2$ such that
\[
  Z_i\,\widetilde D_A\ket{0^a}=(-1)^{\beta_i}\widetilde D_A\ket{0^a}.
\]
The joint eigenspace specified by these $a$ eigenvalue equations is the
one-dimensional space spanned by the computational-basis state
$\ket b$, where $b=(\beta_1,\ldots,\beta_a)$.  Hence there is
$\theta\in\R$ with
\[
  \widetilde D_A\ket{0^{a}}=e^{i\theta}\ket b .
\]
Set $D_A=e^{-i\theta}X^{b}\widetilde D_A$.  Then $D_A\ket{0^{a}}=\ket{0^{a}}$,
which is (i).  Conjugation by the Pauli $X^{b}$ changes only signs of Paulis,
not their binary symplectic vectors, and the global phase changes nothing, so
$D_A$ still induces $S$.

For (ii), $\nu(D_AB_jD_A^{\dagger})=S\nu(B_j)\in\Span\{e_1,f_1,\ldots,e_r,f_r\}$
by \eqref{eq:symp-goal}, and a Pauli whose symplectic vector has vanishing
$e_i$- and $f_i$-coordinates for $i>r$ acts as the identity on those ancillas.
Thus
\begin{equation}\label{eq:compressed-paulis}
  D_AB_jD_A^{\dagger}=\sigma_j\,B_j'\otimes I^{\otimes(a-r)},
  \qquad\sigma_j\in\{\pm1\},
\end{equation}
for some Hermitian Pauli $B_j'$ on $r$ qubits.

Finally, $D_A\ket{0^{a}}=\ket{0^{a}}$ gives
$D(\ket{0^{a}}\otimes I)=\ket{0^{a}}\otimes I$ and
$(\bra{0^{a}}\otimes I)D^{\dagger}=\bra{0^{a}}\otimes I$, whence
\[
  (\bra{0^{a}}\otimes I)\,DUD^{\dagger}\,(\ket{0^{a}}\otimes I)
  =(\bra{0^{a}}\otimes I)\,U\,(\ket{0^{a}}\otimes I)=A_U .
\]
Since $D$ commutes with each $R(P_j)$ only after conjugation, we record for
later use that $DUD^{\dagger}=(DCD^{\dagger})\prod_jR(DP_jD^{\dagger})$, which
is again of the form \eqref{eq:pauli-normal-form} with the same number of
factors.
\end{proof}

\section{Structure of projected Clifford blocks}\label{app:clifford-corner}

Let $C$ be a Clifford on a register $L$ of $q$ qubits together with a register
$F$ of $f$ qubits, and let
\begin{equation}\label{eq:clifford-corner}
  \Lambda=(\bra{0^{q}}_L\otimes I_F)\,C\,(\ket{0^{q}}_L\otimes I_F).
\end{equation}
This appendix proves the following structure theorem.

\begin{theorem}[Projected Clifford block structure]\label{thm:corner-structure}
Assume $\Lambda\ne0$.  Then there exists an integer $\kappa\ge0$, a Clifford $W$
on $F$, and a projector $\Pi$ onto a stabilizer code with $m$ independent
generators, such that
\begin{equation}\label{eq:corner-form}
  \Lambda=2^{-\kappa/2}\,W\,\Pi ,
\end{equation}
with
\begin{equation}\label{eq:corner-ranges}
  m+\kappa\le q,\qquad m\le\min\{q,f\}.
\end{equation}
All nonzero singular values of $\Lambda$ equal $2^{-\kappa/2}$.
\end{theorem}

\subsection{Postselected Pauli measurement}

\begin{lemma}\label{lem:pauli-measurement}
Let $\ket\phi$ be a stabilizer state on $\ell$ qubits with stabilizer group
$\cS$ of rank $\ell$, let $P$ be a Hermitian Pauli operator and
$\Pi_+=\tfrac12(I+P)$.  Exactly one of the following holds:
\begin{enumerate}[label=(\roman*)]
\item $P\in\cS$, and $\Pi_+\ket\phi=\ket\phi$;
\item $-P\in\cS$, and $\Pi_+\ket\phi=0$;
\item $\pm P\notin\cS$, and then $\|\Pi_+\ket\phi\|^{2}=\tfrac12$ and
$\sqrt2\,\Pi_+\ket\phi$ is a stabilizer state whose stabilizer group has rank
$\ell$ and contains $P$.
\end{enumerate}
\end{lemma}

\begin{proof}
Cases (i) and (ii) are mutually exclusive since $-I\notin\cS$, and both follow
from $P\ket\phi=\pm\ket\phi$.

Assume $\pm P\notin\cS$.  A rank-$\ell$ stabilizer group is a maximal abelian
subgroup of the Pauli group not containing $-I$, so its centralizer in the
Pauli group is generated by $\cS$ and phases; consequently $P$ anticommutes
with some element of $\cS$.  Choose generators $g_1,\ldots,g_\ell$ with $g_1$
anticommuting with $P$, and replace $g_j\mapsto g_1g_j$ for every other
anticommuting generator, so that $g_1$ is the unique anticommuting generator.  Using
$g_1\ket\phi=\ket\phi$ and $Pg_1=-g_1P$,
\[
  \bra\phi P\ket\phi=\bra\phi Pg_1\ket\phi=-\bra\phi g_1P\ket\phi=-\bra\phi P\ket\phi ,
\]
so $\bra\phi P\ket\phi=0$ and
$\|\Pi_+\ket\phi\|^{2}=\tfrac12\bigl(1+\bra\phi P\ket\phi\bigr)=\tfrac12$.
For $j\ge2$ the generator $g_j$ commutes with $\Pi_+$, hence
$g_j\Pi_+\ket\phi=\Pi_+\ket\phi$, and $P\Pi_+=\Pi_+$.  Set
\[
  \widetilde{\cS}:=\langle P,g_2,\ldots,g_\ell\rangle .
\]
Then $\widetilde{\cS}$ stabilizes the normalized state
$\sqrt{2}\Pi_+\ket\phi$.  It has rank $\ell$ because
$P\notin\langle g_2,\ldots,g_\ell\rangle$ (the latter commutes with $g_1$,
while $P$ does not), and it omits $-I$ because it stabilizes a nonzero vector.
\end{proof}

\subsection{The Choi state of the projected block}

Let $R$ be a reference register of $f$ qubits and
$\ket{\Omega_f}_{FR}=2^{-f/2}\sum_x\ket x_F\ket x_R$.  Define the unnormalized
Choi vector
\begin{equation}\label{eq:corner-choi}
  \ket J=(\Lambda\otimes I_R)\ket{\Omega_f} .
\end{equation}

\begin{lemma}\label{lem:choi-stab}
There is an integer $0\le t\le q$ with $\|J\|^{2}=2^{-t}$, and
$\ket{\widetilde J}:=2^{t/2}\ket J$ is a stabilizer state on $F\otimes R$.
\end{lemma}

\begin{proof}
Applying $C\otimes I_R$ to $\ket{0^{q}}_L\ket{\Omega_f}_{FR}$ and then
projecting $L$ onto $\ket{0^{q}}$ produces precisely $\ket J$, because
\[
  (\bra{0^{q}}_L\otimes I_{FR})(C\otimes I_R)(\ket{0^{q}}_L\ket{\Omega_f})
  =(\Lambda\otimes I_R)\ket{\Omega_f}.
\]
The input is a stabilizer state: it is stabilized by
$\langle Z_i^{L}\ (i\le q);\ X_j^{F}X_j^{R},\ Z_j^{F}Z_j^{R}\ (j\le f)\rangle$,
of rank $q+2f$.  The Clifford $C\otimes I_R$ preserves this property, and
$\ket{0^{q}}\!\bra{0^{q}}_L=\prod_{i=1}^{q}\tfrac12(I+Z_i)$, so the projection
is a sequence of $q$ postselected Pauli measurements.  By
Lemma~\ref{lem:pauli-measurement} each step multiplies the squared norm by $1$
or $\tfrac12$ and returns a stabilizer state after normalization; the case
$\Pi_+\ket\phi=0$ cannot occur, since $\Lambda\ne0$ implies $\ket J\ne0$.
Letting $t$ be the number of steps in case (iii) of
Lemma~\ref{lem:pauli-measurement} gives $\|J\|^{2}=2^{-t}$ with
$0\le t\le q$.  In the final stabilizer group every $Z_i^{L}$ occurs, so the
register $L$ factorizes in the state $\ket{0^{q}}$, and the remaining $2f$
generators stabilize $\ket{\widetilde J}$ on $F\otimes R$.
\end{proof}

\begin{lemma}[Generalization of the bipartite stabilizer normal form]\label{lem:fattal}
Let $\ket\chi$ be a stabilizer state on registers $F$ and $R$, each consisting
of $f$ qubits.  Then there are Clifford unitaries $C_1$ on $F$ and $C_2$ on
$R$, and integers $p,m\ge0$ with $p+m=f$, such that
\begin{equation}\label{eq:fattal-form}
  \ket\chi=(C_1\otimes C_2)\bigl(\ket{\Omega_p}\otimes\ket{0^{m}}_F\ket{0^{m}}_R\bigr),
\end{equation}
where $\ket{\Omega_p}$ consists of $p$ Bell pairs shared between the first $p$
qubits of $F$ and of $R$.
\end{lemma}

\begin{proof}
Let $\cS$ be the stabilizer group of $\ket\chi$.  The canonical form theorem
for bipartite stabilizer states \cite{fattal2004entanglement} says that, after
choosing suitable independent generators, $\cS$ is generated by local
generators on $F$ and $R$ together with nonlocal pairs
\[
  g_i^F\otimes g_i^R,\qquad \bar g_i^F\otimes \bar g_i^R,
  \qquad 1\le i\le p,
\]
where $g_i^F,\bar g_i^F$ form $p$ anticommuting Pauli pairs on $F$, and the
same is true on $R$.  The local stabilizer subgroups on the two sides have the
same rank $m$, and the canonical form gives $p+m=f$.

On the $F$ side, the $2p+m$ Pauli operators consisting of the $p$
anticommuting pairs and the $m$ local stabilizer generators have the same
commutation relations as
\[
  X_i,\ Z_i\quad (1\le i\le p),\qquad
  Z_{p+j}\quad (1\le j\le m).
\]
The Clifford--symplectic correspondence therefore gives a local Clifford
taking one family to the other, up to signs; these signs can be corrected by
local Paulis.  Applying the same argument on $R$, we obtain local Clifford
unitaries $C_1$ and $C_2$ such that the stabilizer group is mapped to the
standard group generated by
\[
  X_i^F X_i^R,\ Z_i^F Z_i^R\quad (1\le i\le p),
  \qquad
  Z_{p+j}^F,\ Z_{p+j}^R\quad (1\le j\le m),
\]
up to conjugating $\ket\chi$ by $C_1^\dagger\otimes C_2^\dagger$.

This standard group stabilizes
$\ket{\Omega_p}\otimes\ket{0^m}_F\ket{0^m}_R$.  Since a stabilizer state is
determined by its stabilizer group up to global phase, conjugating back gives
\eqref{eq:fattal-form} up to an overall phase; absorbing that phase into
$C_1$ gives the displayed equality.
\end{proof}

\begin{proof}[Proof of Theorem~\ref{thm:corner-structure}]
By Lemma~\ref{lem:choi-stab}, $\ket{\widetilde J}=2^{t/2}\ket J$ is a stabilizer
state on $F\otimes R$; apply Lemma~\ref{lem:fattal} to it, obtaining $C_1,C_2$
and $p+m=f$ with \eqref{eq:fattal-form}.

Let $\Pi_0=I_{2^{p}}\otimes\proj{0^{m}}$ on $f$ qubits.  We have
\[
  (\Pi_0\otimes\Pi_0)\ket{\Omega_f}=2^{-m/2}\,\ket{\Omega_p}\otimes\ket{0^{m}}\ket{0^{m}},
\]
where the right-hand side is written after reordering the qubits so that
$\ket{\Omega_f}=\ket{\Omega_p}\otimes\ket{\Omega_m}$ and using
$(\proj{0^{m}}\otimes\proj{0^{m}})\ket{\Omega_m}=2^{-m/2}\ket{0^{m}}\ket{0^{m}}$.
Combining this with \eqref{eq:fattal-form} and $\ket J=2^{-t/2}\ket{\widetilde J}$,
\[
  \ket J=2^{-(t-m)/2}\,(C_1\Pi_0\otimes C_2\Pi_0)\ket{\Omega_f}.
\]
Using the transpose trick $(I\otimes M)\ket{\Omega_f}=(M^{\top}\otimes I)\ket{\Omega_f}$
together with $\Pi_0^{\top}=\Pi_0=\Pi_0^{2}$,
\[
  \ket J=2^{-(t-m)/2}\,\bigl(C_1\Pi_0C_2^{\top}\otimes I\bigr)\ket{\Omega_f}.
\]
The map $M\mapsto(M\otimes I)\ket{\Omega_f}$ is injective, so comparison with
\eqref{eq:corner-choi} gives
\begin{equation}\label{eq:Lambda-explicit}
  \Lambda=2^{-(t-m)/2}\,C_1\Pi_0C_2^{\top} .
\end{equation}

The transpose of a Clifford is Clifford: for unitary $V$ the matrix $V^{\top}$
is unitary and $V^{\top}PV^{\top\dagger}=(V^{\dagger}P^{\top}V)^{\top}$, while
$P^{\top}=\pm P$ for any Pauli $P$; hence $V$ Clifford implies $V^{\top}$
Clifford.  Setting
\[
  W=C_1C_2^{\top},\qquad \Pi=(C_2^{\top})^{\dagger}\Pi_0C_2^{\top},
  \qquad \kappa=t-m,
\]
Eq.~\eqref{eq:Lambda-explicit} becomes \eqref{eq:corner-form}, with $W$
Clifford.  Moreover $\Pi_0=\prod_{i=p+1}^{f}\tfrac12(I+Z_i)$ is the projector
onto the stabilizer code with generators $Z_{p+1},\ldots,Z_f$, and conjugation
by the Clifford $C_2^{\top}$ maps these to $m$ independent commuting Hermitian
Paulis generating a group that does not contain $-I$; hence $\Pi$ is a
stabilizer-code projector with $m$ independent generators.

It remains to verify \eqref{eq:corner-ranges}.  Since $W$ is unitary and $\Pi$
an orthogonal projector, all nonzero singular values of $\Lambda$ equal
$2^{-\kappa/2}$.  As $\Lambda$ is a block of the unitary $C$ we have
$\|\Lambda\|\le1$, so $2^{-\kappa/2}\le1$ and $\kappa\ge0$, i.e.\ $t\ge m$.
Together with $t\le q$ from Lemma~\ref{lem:choi-stab} this gives
$m+\kappa=t\le q$, hence $m\le q$; and $m\le f$ because $p=f-m\ge0$.
\end{proof}

\section{Completing the ancilla-compression proof}
\label{app:compressed-circuit}

We first give the coherent reconstruction used in the proof of
Theorem~\ref{thm:compression}.

\begin{proposition}[Coherent reconstruction]
\label{prop:reconstruct}
Let the ancilla register be divided as $A=M\sqcup B$, where
$|M|=s$ and $|B|=q$, and let $S$ be an $n$-qubit system register.
Set
\[
    F:=M\sqcup S.
\]
Suppose an ancilla Clifford $D_A$ satisfies
$D_A\ket{0^{s+q}}=\ket{0^{s+q}}$, and that, with
$D=D_A\otimes I_S$,
\[
    DUD^\dagger=C'R_F,
    \qquad
    R_F=\prod_{j=1}^{s}R(P_j'),
\]
where $C'$ is Clifford and every $P_j'$ acts only on $F$.

Let $\widetilde H:=\alpha A_U$ and define
\[
    \Lambda
    :=
    (\bra{0^q}_B\otimes I_F)
    C'
    (\ket{0^q}_B\otimes I_F).
\]
Then the following statements hold.

\begin{enumerate}[label=(\roman*),leftmargin=*]

\item The encoded block satisfies
\begin{equation}
\label{eq:reduced-block}
    \frac{\widetilde H}{\alpha}
    =
    (\bra{0^s}_M\otimes I_S)
    \Lambda R_F
    (\ket{0^s}_M\otimes I_S).
\end{equation}
In particular, if $\widetilde H\ne0$, then $\Lambda\ne0$.

\item Suppose
\[
    \Lambda=2^{-\kappa/2}W\Pi
\]
as in Theorem~\ref{thm:corner-structure}, where
\[
    \Pi=\prod_{j=1}^{m}\frac{I+g_j}{2}
\]
is a stabilizer-code projector with $m$ independent generators.
Introduce a syndrome register $E$ consisting of $m$ qubits
$e_1,\ldots,e_m$, and define
\begin{equation}
\label{eq:projector-gadget}
    G_j
    :=
    H_{e_j}
    \bigl(
        \proj0_{e_j}\otimes I
        +
        \proj1_{e_j}\otimes g_j
    \bigr)
    H_{e_j}.
\end{equation}
Then
\[
    U'
    :=
    W\left(\prod_{j=1}^{m}G_j\right)R_F
\]
is a unitary Clifford$+T$ circuit on $M\sqcup E\sqcup S$ with
$T(U')\le s$, and
\[
    (\bra{0^{s+m}}_{ME}\otimes I_S)
    U'
    (\ket{0^{s+m}}_{ME}\otimes I_S)
    =
    \frac{\widetilde H}{\alpha'},
    \qquad
    \alpha'=2^{-\kappa/2}\alpha.
\]

\end{enumerate}
\end{proposition}

\begin{proof}
Since $D_A$ fixes $\ket{0^{s+q}}$, conjugation by $D$ preserves the
encoded block. Hence
\[
    \frac{\widetilde H}{\alpha}
    =
    (\bra{0^{s+q}}\otimes I_S)
    C'R_F
    (\ket{0^{s+q}}\otimes I_S).
\]
Writing
\[
    \ket{0^{s+q}}
    =
    \ket{0^s}_M\ket{0^q}_B
\]
and using the fact that $R_F$ acts trivially on $B$ gives
Eq.~\eqref{eq:reduced-block}. If $\Lambda=0$, its right-hand side
vanishes, so $\widetilde H=0$. This proves part~(i).

For part~(ii), each $g_j$ is a Pauli operator, so the controlled
operation in Eq.~\eqref{eq:projector-gadget} is Clifford. Therefore
each $G_j$ is Clifford. Moreover,
\[
    \bra0_{e_j}G_j\ket0_{e_j}
    =
    \bra+_{e_j}
    \bigl(
        \proj0_{e_j}\otimes I
        +
        \proj1_{e_j}\otimes g_j
    \bigr)
    \ket+_{e_j}
    =
    \frac{I+g_j}{2}.
\]
Since the generators $g_j$ commute and each gadget uses a different
syndrome qubit,
\[
    \bra{0^m}_E
    \left(\prod_{j=1}^{m}G_j\right)
    \ket{0^m}_E
    =
    \prod_{j=1}^{m}\frac{I+g_j}{2}
    =
    \Pi.
\]
It follows that
\begin{align*}
    &(\bra{0^{s+m}}_{ME}\otimes I_S)
    U'
    (\ket{0^{s+m}}_{ME}\otimes I_S)
    \\
    &\quad=
    (\bra{0^s}_M\otimes I_S)
    W\Pi R_F
    (\ket{0^s}_M\otimes I_S)
    \\
    &\quad=
    2^{\kappa/2}
    (\bra{0^s}_M\otimes I_S)
    \Lambda R_F
    (\ket{0^s}_M\otimes I_S)
    =
    \frac{\widetilde H}{\alpha'},
\end{align*}
where the last equality uses Eq.~\eqref{eq:reduced-block} and
$\alpha'=2^{-\kappa/2}\alpha$.

Finally, $W$ and all $G_j$ are Clifford. The only non-Clifford part of
$U'$ is therefore $R_F$, which contains at most $s$ Pauli
$\pi/8$ rotations. Hence $T(U')\le s$.
\end{proof}

\begin{proof}[Proof of Theorem~\ref{thm:compression}]
If
\[
    a\le n+2s,
\]
we retain the original circuit. Taking
\[
    U'=U,
    \qquad
    a'=a,
    \qquad
    \kappa=0,
    \qquad
    \alpha'=\alpha
\]
gives all the claimed bounds. We therefore assume
\[
    a>n+2s.
\]
In particular, $a>s$.

Let
\[
    \widetilde H:=\alpha A_U.
\]
Since $U$ is an $(\alpha,a,\eps)$ block encoding of $H$,
\[
    \|H-\widetilde H\|\le\eps.
\]
Thus, it is enough to compress the exact block encoding of
$\widetilde H$.

If $\widetilde H=0$, then $\|H\|\le\eps$. The Clifford circuit
$X\otimes I_S$, with one block-encoding ancilla, has zero encoded
block. It is therefore an $(\alpha,1,\eps)$ block encoding of $H$
with no $T$ gates. Since $n\ge1$ in the Hamiltonian families
considered here,
\[
    1\le n+2s<a,
\]
and the theorem follows with $\kappa=0$. We henceforth assume
$\widetilde H\ne0$.

Write $U$ in the Pauli-rotation normal form
\eqref{eq:pauli-normal-form}:
\[
    U=C\prod_{j=1}^{s}R(P_j).
\]
Decompose each Pauli generator as
\[
    P_j=B_j\otimes Q_j,
\]
where $B_j$ acts on the ancilla register and $Q_j$ acts on the system.
Set
\[
    r
    :=
    \dim_{\F_2}
    \Span\{\nu(B_1),\ldots,\nu(B_s)\}
    \le s.
\]

By Proposition~\ref{prop:pauli-compression}, there is an ancilla
Clifford $D_A$ satisfying
\[
    D_A\ket{0^a}=\ket{0^a}
\]
such that every $D_AB_jD_A^\dagger$ is supported on the same set of
at most $r$ ancillas. Since $r\le s$, we may enlarge this set to a
register $M$ of exactly $s$ ancillas. Let $B$ contain the remaining
\[
    q:=a-s
\]
ancillas, and set
\[
    F:=M\sqcup S,
    \qquad
    |F|=n+s.
\]

With $D=D_A\otimes I_S$, the conjugated circuit has the form
\[
    DUD^\dagger=C'R_F,
    \qquad
    R_F=\prod_{j=1}^{s}R(P_j'),
\]
where $C'$ is Clifford and every $P_j'$ acts only on $F$.

Define
\[
    \Lambda
    :=
    (\bra{0^q}_B\otimes I_F)
    C'
    (\ket{0^q}_B\otimes I_F).
\]
By Proposition~\ref{prop:reconstruct}(i),
\[
    \frac{\widetilde H}{\alpha}
    =
    (\bra{0^s}_M\otimes I_S)
    \Lambda R_F
    (\ket{0^s}_M\otimes I_S).
\]
Since $\widetilde H\ne0$, we have $\Lambda\ne0$.

Theorem~\ref{thm:corner-structure} therefore gives
\[
    \Lambda=2^{-\kappa/2}W\Pi,
\]
where $\kappa\ge0$, $W$ is Clifford, and $\Pi$ is a stabilizer-code
projector with $m$ independent generators satisfying
\[
    m\le\min\{q,n+s\}\le n+s.
\]

Applying Proposition~\ref{prop:reconstruct}(ii) gives a unitary
circuit $U'$ with
\[
    T(U')\le s,
    \qquad
    a'=s+m\le n+2s,
\]
and
\[
    \alpha'A_{U'}=\widetilde H,
    \qquad
    \alpha'=2^{-\kappa/2}\alpha.
\]
Consequently,
\[
    \|H-\alpha'A_{U'}\|
    =
    \|H-\widetilde H\|
    \le\eps,
\]
so $U'$ is an $(\alpha',a',\eps)$ block encoding of $H$.

Because we are in the case $a>n+2s$, the compressed circuit satisfies $a'\le n+2s<a$.
Combining this with the first case gives
\[
    a'\le\min\{a,n+2s\}.
\]
Finally, $\kappa\ge0$ implies
\[
    \alpha'=2^{-\kappa/2}\alpha\le\alpha.
\]
This proves the theorem.
\end{proof}

\section{Coefficient estimates for second-quantized Hamiltonians}\label{app:second-coefficient-norm}
\begin{proof}[Proof of Lemma~\ref{lem:coefficient-norm}]
Let $N_0=\|H(h,V)\|$ and let $\ket\Omega$ be the vacuum.  We write
$x=(h,V)$ and use the elementary fact that the Euclidean norm of the
independent real coordinates of a Hermitian matrix is at most its Frobenius
norm.

\emph{Step 1: one-body coefficients.}
Put $\ket p:=a_p^{\dagger}\ket\Omega$.  Since the two-body term annihilates the
one-particle sector,
\begin{equation}
  \bra p H(h,V)\ket q=h_{pq}.
  \label{eq:one-body-matrix-element}
\end{equation}
Thus $|h_{pq}|\le N_0$ for every $p,q$.  The restriction of $H(h,V)$ to the
one-particle sector is the matrix $h$, so $\|h\|_{\mathrm{op}}\le N_0$.
Since $h$ is an $n\times n$ matrix, its Frobenius norm is bounded by
$\sqrt n$ times its operator norm:
\begin{equation}
  \|h\|_F\le \sqrt n\,N_0.
  \label{eq:hbound}
\end{equation}

\emph{Step 2: two-body coefficients.}
For $P=(p,q)$ with $p<q$, write $\ket P:=b_P^\dagger\ket\Omega$. Fix
$P=(p,q)$ and $R=(r,s)$ with $p<q$ and $r<s$.  By direct computation, the two-particle matrix
element is
\begin{equation}
  \bra P H(h,V)\ket R
  =
  V_{PR}
  +\delta_{qs}h_{pr}-\delta_{ps}h_{qr}
  -\delta_{qr}h_{ps}+\delta_{pr}h_{qs}.
  \label{eq:two-body-matrix-element}
\end{equation}
Indeed, the $V$-term contributes $V_{PR}$ because $b_R\ket R=\ket\Omega$,
and the displayed Kronecker-delta expression is the contribution of
$\sum_{a,b}h_{ab}a_a^\dagger a_b$.  For ordered pairs $p<q$ and $r<s$, at most
two of the four delta terms in \eqref{eq:two-body-matrix-element} are nonzero.
Together with \eqref{eq:one-body-matrix-element} and
$|\bra{P}H(h,V)\ket{R}|\le N_0$, this gives
\begin{equation}
  |V_{PR}|\le 3N_0
  \qquad\text{for every }P,R.
  \label{eq:Vcoef}
\end{equation}
Moreover, when $|P\cup R|=4$ all four delta terms vanish, so
$\bra P H(h,V)\ket R=V_{PR}$.  Let $V^{(4)}$ be the subvector of independent
real two-body coordinates with $|P\cup R|=4$.  If $\Pi_2$ denotes the projection
onto the two-particle sector, then $\|\Pi_2 H\Pi_2\|\le N_0$ and
$\operatorname{rank}(\Pi_2 H\Pi_2)\le D=\binom n2$, hence
\begin{equation}
  \sum_{|P\cup R|=4}|V_{PR}|^2
  \le \|\Pi_2 H\Pi_2\|_F^2
  \le D N_0^2,
  \label{eq:bulk}
\end{equation}
where the first sum is over ordered pairs $P,R\in\cI$.  Since each independent
off-diagonal coordinate is counted twice in this ordered sum,
\begin{equation}
  \|V^{(4)}\|_2^2\le \frac{D}{2}N_0^2\le \frac{n^2}{4}N_0^2.
  \label{eq:bulk-norm}
\end{equation}

\emph{Step 3: the remaining coordinates.}
We now separate the full coefficient vector as
\[
  x=(V^{(4)},x_{\le3}),
\]
where $x_{\le3}$ contains all independent real coordinates of the one-body
matrix $h$ and all independent real coordinates $V_{PR}$ with
$|P\cup R|\le3$.  
The number of real coordinates in $x_{\le3}$ is
\[
  n^2+D^2-6\binom n4
  =
  n^3-\frac32 n^2+\frac32 n
  \le n^3
  \qquad(n\ge2).
\]
Here $6\binom n4$ is precisely the number of real two-body coordinates already
placed in $V^{(4)}$: for each four-element set there are three partitions into
two unordered pairs, and each corresponding off-diagonal Hermitian entry
contributes its real and imaginary parts.

By \eqref{eq:one-body-matrix-element}, every real one-body coordinate in
$x_{\le3}$ has modulus at most $N_0$; by \eqref{eq:Vcoef}, every two-body
coordinate in $x_{\le3}$ has modulus at most $3N_0$.  Therefore every real
coordinate of $x_{\le3}$ has modulus at most $3N_0$, and
\begin{equation}
  \|x_{\le3}\|_2^2\le 9n^3N_0^2.
  \label{eq:rest}
\end{equation}

Combining \eqref{eq:bulk-norm} and \eqref{eq:rest},
\[
  \|x\|_2^2=\|V^{(4)}\|_2^2+\|x_{\le3}\|_2^2
  \le \frac{n^2}{4}N_0^2+9n^3N_0^2
  \le 10n^3N_0^2.
\]
Taking square roots gives
\begin{equation}
  \|x\|_2\le 4n^{3/2}N_0,
\end{equation}
which is precisely \eqref{eq:norm-comparison}.
\end{proof}

\section{Covering estimates}\label{app:second-covering-est}
\begin{proof}[Proof of Lemma~\ref{lem:tube}]
Suppose first that $A_U=0$.  Then $x\in S_{[U]}$ forces $\|H(x)\|\le\eps$, so
Lemma~\ref{lem:coefficient-norm}, applied outside any bounded-coefficient
restriction, gives $\|x\|_2\le\eps/c_0\le\rho$; thus $S_{[U]}$
lies in the ball $B_d(0,\rho)$, which is the degenerate case $x_\star=0$ of
\eqref{eq:tube-inclusion}.  For the volume bound, enclose this ball in the
cylinder $B_{d-1}(0,\rho)\times[-\rho,\rho]$ to get
\[
  \vol(S_{[U]})\le 2\rho\,\omega_{d-1}\rho^{d-1}
  \le (4\sqrt d+3\rho)\omega_{d-1}\rho^{d-1}.
\]
Now suppose $A_U\ne0$.

\emph{Choice of $x_\star$.}  Put
\[
  \mathcal Z_U
  =
  \bigl\{z\in\C\setminus\{0\}:
    \exists x\in\cK,\ \|H(x)-zA_U\|\le\eps
  \bigr\},
\]
which is nonempty by hypothesis.  For $z\in\mathcal Z_U$ with witness $x$,
the reverse triangle inequality and \eqref{eq:Gamma} give
\[
  |z|\|A_U\|=\|zA_U\|
  \le\|H(x)\|+\eps
  \le\Gamma+\eps,
\]
so $\sup_{z\in\mathcal Z_U}|z| \le(\Gamma+\eps)/\|A_U\|<\infty$. Choose $z_\star\in\mathcal Z_U$ such that
\[
  |z_\star|
  \ge
  \frac12\sup_{z\in\mathcal Z_U}|z|,
\] and fix a corresponding witness $x_\star\in S_{[U]}$ for it.
Every witness $z$ of every $x\in S_{[U]}$ then satisfies
\begin{equation}\label{eq:mu-bound}
\mu:=\frac{z}{z_\star},\qquad
  |\mu|\le 2 .
\end{equation}

\emph{The tube.}  Let $x\in S_{[U]}$ have witness $z$, and let $\mu$ be as in
\eqref{eq:mu-bound}.  Using
$z=\mu z_\star$,
\begin{equation}\label{eq:elim}
  \|H(x)-\mu H(x_\star)\|
  \le (1+|\mu|)\eps
  \le 3\eps .
\end{equation}
Since $H(x)$ and $H(x_\star)$ are Hermitian, taking the Hermitian
part gives
\[
  \|H(x)-\operatorname{Re}(\mu)H(x_\star)\|
  \le 3\eps .
\]
By Lemma~\ref{lem:coefficient-norm},
\[
  \|x-\operatorname{Re}(\mu)x_\star\|_2
  \le \frac{3\eps}{c_0}
  \le \rho .
\]
Finally, $|\operatorname{Re}(\mu)|\le|\mu|\le2$, so
$\operatorname{Re}(\mu)x_\star$ lies on the segment
$\{\nu x_\star:-2\le\nu\le2\}$.  This proves
\eqref{eq:tube-inclusion}.

\emph{Volume.}  By \eqref{eq:K-volume} the segment has length
$4\|x_\star\|_2\le4\sqrt d$.  A tube of radius $\rho$ about this segment of length
$4\|x_\star\|_2$ is contained in the union of a cylinder and two half-balls, so its
volume is at most $4\sqrt d\,\omega_{d-1}\rho^{d-1}+\omega_d\rho^{d}$.  To compare
the second term with the cylinder term, write
\[
  \frac{\omega_d}{\omega_{d-1}}
  =
  \sqrt\pi\,
  \frac{\Gamma((d+1)/2)}{\Gamma(d/2+1)} .
\]
Gautschi's inequality, in the form
$\Gamma(z+1)/\Gamma(z+1/2)\ge\sqrt z$ for $z>0$, applied with $z=d/2$, gives
\[
  \frac{\omega_d}{\omega_{d-1}}
  \le \sqrt{\frac{2\pi}{d}}
  \le 3
  \qquad(d\ge1).
\]
Thus the tube volume is bounded by
$(4\sqrt d+3\rho)\omega_{d-1}\rho^{d-1}$, proving
\eqref{eq:tube-volume}.
\end{proof}

\begin{proof}[Proof of Proposition~\ref{prop:covering}]
By hypothesis and Lemma~\ref{lem:tube},
$M\cdot\max_{[U]}\vol(S_{[U]})\ge\vol(\cK)\ge2^{d/2}$, so
\begin{equation}\label{eq:Mchain}
  M\ \ge\ \frac{2^{d/2}}{(4\sqrt d+3\rho)\,\omega_{d-1}\rho^{\,d-1}} .
\end{equation}
From $\Gamma(k/2+1)\ge(k/2e)^{k/2}$ we get
$\omega_k=\pi^{k/2}/\Gamma(k/2+1)\le(2\pi e/k)^{k/2}$, hence
\[
  \omega_{d-1}\rho^{\,d-1}\le\Bigl(\frac{\sqrt{2\pi e}\,\rho}{\sqrt{d-1}}\Bigr)^{d-1}.
\]
Substituting into Eq.~\eqref{eq:Mchain} gives,
\begin{equation}\label{eq:Mchain2}
  M\ \ge\ \frac{\sqrt2}{4\sqrt d+3\rho}
   \left(\frac{\sqrt{d-1}}{\sqrt{\pi e}\,\rho}\right)^{d-1}.
\end{equation}
With $\rho=16n^{3/2}\eps$ and $d-1\ge n^{4}/8$ for $n\ge2$, the base satisfies
\[
  \frac{\sqrt{d-1}}{\sqrt{\pi e}\,\rho}
  \ \ge\ \frac{n^{2}/(2\sqrt2)}{\sqrt{\pi e}\cdot16n^{3/2}\eps}
  \ =\ \frac{c_\star\,\sqrt n}{\eps},
\]
where $c_\star:=\bigl(32\,\sqrt{2\pi e}\bigr)^{-1}$.  Finally $\eps\le1$ gives
$\rho\le16n^{3/2}$.  Since $d\ge n^4/8$,
\[
  32\sqrt d\ge 8\sqrt2\,n^2\ge16n^{3/2}
  \qquad(n\ge2),
\]
and therefore
\[
  4\sqrt d+3\rho\le100\sqrt d.
\]
Since $d\ge5$, this gives
\[
  \log_2(4\sqrt d+3\rho)
  \le \log_2 100+\frac12\log_2 d
  \le 4\log_2 d.
\]
Taking logarithms in \eqref{eq:Mchain2} yields \eqref{eq:covering-count}.
\end{proof}

\section{Details for the Kitaev constant-error bound}\label{app:kitaev-gap}

We first establish the Pauli-eigenspace gap needed for
part~\ref{kitaev-part2}, and then complete its proof. Recall
$M=\tfrac32n$ and, from Lemma~\ref{lem:bond-action} with all $J_e=1$,
\begin{equation}\label{eq:uniform-state-explicit}
  \ket\Psi=\frac{1}{\sqrt M}\Bigl(
    i\sum_{(u,v)\in E_x}\ket{e_v}
    -i\sum_{(u,v)\in E_y}\ket{e_u}
    +\sum_{(u,v)\in E_z}\ket{e_u+e_v}\Bigr),
\end{equation}
a unit vector supported on the $M$ distinct basis states of $\cS$ in
\eqref{eq:support}.

\begin{lemma}\label{lem:pauli-gap}
Let $n\ge4$.  For every Hermitian Pauli $P\ne\pm I$ and either choice of sign,
\[
  \dist\bigl(\ket\Psi,\ker(P\mp I)\bigr)\ \ge\ \frac{1}{\sqrt{2M}} .
\]
\end{lemma}

\begin{proof}
Write $\ket\Psi=\sum_{s\in\cS}c_s\ket s$ with $c_s\in M^{-1/2}\{1,i,-i\}$.  In
the local product basis every entry of a Pauli operator lies in
$\{0,\pm1,\pm i\}$.  Writing
\[
  \bra\Psi P\ket\Psi
  =
  \sum_{s,s'\in\cS}\overline{c_{s'}}\,P_{s's}\,c_s,
  \qquad
  P_{s's}:=\bra{s'}P\ket s,
\]
each summand indexed by an ordered pair $(s',s)\in\cS\times\cS$ lies in
$M^{-1}\{0,\pm1,\pm i\}$.  Hence
$\bra\Psi P\ket\Psi\in M^{-1}\Z[i]$, and it is real because $P$ is Hermitian;
therefore $\bra\Psi P\ket\Psi\in M^{-1}\Z$.

By Lemma~\ref{lem:kitaev-trivial-stab} applied with all $J_e=1$, no Hermitian
Pauli other than $I$ stabilizes $\ket\Psi$.  Thus $P\ket\Psi\ne\ket\Psi$.
For a Hermitian Pauli, $\bra\Psi P\ket\Psi=1$ would force
\[
  0=1-\bra\Psi P\ket\Psi
  =2\Bigl\|\frac{I-P}{2}\ket\Psi\Bigr\|^2,
\]
and hence $P\ket\Psi=\ket\Psi$, so $\bra\Psi P\ket\Psi\ne1$.  Also
$\bra\Psi P\ket\Psi\ne-1$, because otherwise the same argument applied to $-P$
would give $(-P)\ket\Psi=\ket\Psi$, with $-P\ne\pm I$.  Combined with
$|\bra\Psi P\ket\Psi|\le1$ and the discreteness just established,
\[
  |\bra\Psi P\ket\Psi|\ \le\ 1-\frac1M .
\]
It remains to translate this expectation bound into distance from the
eigenspaces.  We do this for $\ker(I+P)$; the case $\ker(I-P)$ is identical.
Let $\Pi_+=\tfrac12(I+P)$.  Since $\Pi_+$ projects onto the orthogonal
complement of $\ker(I+P)$,
\[
  \dist\bigl(\ket\Psi,\ker(I+P)\bigr)^2
  =
  \|\Pi_+\ket\Psi\|^2
  =
  \frac{1+\bra\Psi P\ket\Psi}{2}
  \ge
  \frac{1}{2M},
\]
which proves the desired bound for $\ker(I+P)$.  Replacing $P$ by $-P$ gives the
same bound for $\ker(I-P)$.
\end{proof}

\begin{proof}[Proof of Theorem~\ref{thm:kitaev-size},
part~\ref{kitaev-part2}]
 By \eqref{eq:support} with all $J_e=1$ we have
$\|H_{\mathrm{unif}}\ket\Phi\|=\sqrt M$; write
$\ket\Psi=M^{-1/2}H_{\mathrm{unif}}\ket\Phi$.
Lemma~\ref{lem:pauli-gap} shows that for every
Hermitian Pauli $P\ne\pm I$ and either sign,
\begin{equation}\label{eq:pauli-eigenspace-gap}
  \dist\bigl(\ket\Psi,\ker(P\mp I)\bigr)\ \ge\ \frac{1}{\sqrt{2M}} .
\end{equation}

Let $U$ be an $(\alpha,a,\eps)$ block encoding of $H_{\mathrm{unif}}$.  From
$\|H_{\mathrm{unif}}-\alpha A_U\|\le\eps$ we get
$\|H_{\mathrm{unif}}\ket\Phi-\alpha A_U\ket\Phi\|\le\eps$.  Since
$\|H_{\mathrm{unif}}\ket\Phi\|=\sqrt M$ and $\eps\le1/4<\sqrt M$, the vector
$A_U\ket\Phi$ is nonzero.  Define the normalized postselected system state
\[
  \ket{\widehat\psi}:=\frac{A_U\ket\Phi}{\|A_U\ket\Phi\|}.
\]
Equivalently, postselecting the block ancillas of
$U(\ket{0^a}\ket\Phi)$ onto $\ket{0^a}$ gives
$\ket{0^a}\otimes\ket{\widehat\psi}$ after normalization.  For any nonzero
vectors $w,w'$ with
$\|w-w'\|\le\eps$ one has
$\bigl\|w/\|w\|-w'/\|w'\|\bigr\|\le2\eps/\max(\|w\|,\|w'\|)$; hence the
state $\ket{\widehat\psi}$ satisfies
$\|\ket{\widehat\psi}-\ket\Psi\|\le2\eps/\sqrt M$.

If $\eps\le1/4$ then
\[
  \|\ket{\widehat\psi}-\ket\Psi\|
  \le \frac{2\eps}{\sqrt M}
  \le \frac{1}{2\sqrt M}
  <
  \frac{1}{\sqrt{2M}} .
\]
Suppose, for contradiction, that $\ket{\widehat\psi}\in\ker(P\mp I)$ for some
Hermitian Pauli $P\ne\pm I$.  Then \eqref{eq:pauli-eigenspace-gap} gives
\[
  \|\ket{\widehat\psi}-\ket\Psi\|
  \ge
  \dist\bigl(\ket\Psi,\ker(P\mp I)\bigr)
  \ge
  \frac{1}{\sqrt{2M}},
\]
contradicting the preceding display.  Thus $\ket{\widehat\psi}$ lies in no
eigenspace of any Hermitian Pauli $P\ne\pm I$; equivalently, it has no
nontrivial Pauli stabilizer and $\nu(\ket{\widehat\psi})=n$.  As in part~\ref{kitaev-part1}, this
forces $T(U)\ge n$.
\end{proof}

\section{Proof of the Kitaev precision lower bound}\label{app:kitaev-precision}
\begin{proof}[Proof of Theorem~\ref{thm:kitaev-precision}]
Pick a vertex $u$ and let $(u,v_x)\in E_x$ and $(u,v_z)\in E_z$ be its $x$- and
$z$-bonds; then $v_x\ne v_z$.  Put
\[
  P=X_uX_{v_x},\qquad Q=Z_uZ_{v_z} .
\]
These are the bond operators of the two chosen edges, and they anticommute
because they differ at the single common vertex $u$, where $X$ and $Z$
anticommute, and act on disjoint sets elsewhere.  With $p:=8\eps\in(0,1/4)$ set
\begin{equation}\label{eq:precision-H}
  H_\eps=\sqrt p\,P+\sqrt{1-p}\,Q ,
\end{equation}
which belongs to $\cK_n$ since $\sqrt p,\sqrt{1-p}\in[0,1]$ and all other
couplings vanish.  The operators $P$ and $Q$ are Hermitian Paulis, so
$H_\eps$ is Hermitian.  From $P^{2}=Q^{2}=I$ and $PQ=-QP$ we get
$H_\eps^{2}=I$; hence the spectrum of $H_\eps$ lies in $\{\pm1\}$, and
$\|H_\eps\|=1$.

Consider the stabilizer state
$\ket\Phi=\ket+_u\otimes\ket R$ with
$\ket R=\ket+_{v_x}\otimes\ket0_{v_z}\otimes\ket{0\cdots0}_{\text{rest}}$.
Then $X_u\ket+_u=\ket+_u$, $X_{v_x}\ket+_{v_x}=\ket+_{v_x}$, and
$Z_{v_z}\ket0_{v_z}=\ket0_{v_z}$, so
\[
  P\ket\Phi=\ket+_u\ket R,\qquad Q\ket\Phi=\ket-_u\ket R,
\]
whence
\[
  H_\eps\ket\Phi=\bigl(\sqrt p\ket+ +\sqrt{1-p}\ket-\bigr)_u\otimes\ket R,
  \qquad \|H_\eps\ket\Phi\|=1 .
\]
Applying a Hadamard on $u$ turns the first factor into
$\ket{\psi_p}=\sqrt p\ket0+\sqrt{1-p}\ket1$, so up to a Clifford and a tensor
factor which is a stabilizer state, the normalized target is $\ket{\psi_p}$.

Now let $U$ be an $(\alpha,a,\eps)$ block encoding of $H_\eps$.  As in the
proof of part~\ref{kitaev-part2} of Theorem~\ref{thm:kitaev-size},
$\|H_\eps\ket\Phi-\alpha A_U\ket\Phi\|\le\eps$ and $\|H_\eps\ket\Phi\|=1$, so
the normalized postselected system output, denoted $\ket{\widehat\Phi}$, is
within Euclidean distance $2\eps$ of $H_\eps\ket\Phi$.  Let
\[
  \ket\eta
  :=
  \bigl(\sqrt p\ket+ +\sqrt{1-p}\ket-\bigr)_u\otimes\ket R
  =
  H_\eps\ket\Phi
\]
and set
\[
  \widehat\rho_u
  :=
  H_u\,\Tr_{\overline u}\bigl(\proj{\widehat\Phi}\bigr)\,H_u ,
\]
where $\Tr_{\overline u}$ traces out all qubits except $u$.  Since partial trace
and unitary conjugation do not increase trace distance,
\[
  \begin{aligned}
  \frac12\bigl\|\widehat\rho_u-\proj{\psi_p}\bigr\|_1
  &=
  \frac12\Bigl\|
  H_u\Tr_{\overline u}\bigl(\proj{\widehat\Phi}\bigr)H_u
  -
  H_u\Tr_{\overline u}\bigl(\proj{\eta}\bigr)H_u
  \Bigr\|_1 \\
  &\le
  \frac12\bigl\|\proj{\widehat\Phi}-\proj{\eta}\bigr\|_1 \\
  &\le
  \|\ket{\widehat\Phi}-\ket\eta\|
  \le
  2\eps
  \le \delta ,
  \end{aligned}
\]
where $\delta:=4\eps$ and the second inequality uses
$\tfrac12\|\proj{\phi}-\proj{\phi'}\|_1
\le\|\ket\phi-\ket{\phi'}\|$ for unit vectors.

The hypotheses of Fact~\ref{fact:bchk} hold: $\delta=4\eps<1/8$ because
$\eps<1/32$, and $p=8\eps=2\delta$ lies strictly between $\delta$ and
$3\delta$.  Replacing every $T$ gate of $U$ by the standard magic-state
injection gadget \cite{bravyi2005magic} converts the block-encoding procedure
followed by postselection into a stabilizer protocol consuming at most $T(U)$
copies of $\ket T$.  Fact~\ref{fact:bchk} then gives
$T(U)\ge\tfrac13\log_2(1/\delta)-\tfrac23$, which is
\eqref{eq:kitaev-precision-bound}.
\end{proof}

\section{Implementation details for the Kitaev block encoding}
\label{app:kitaev-blockencoding}
We give the implementation and error estimates used in Section~\ref{subsec:kitaev-upper}.
\paragraph{Cost of SELECT.}
Let $F$ be a flag register indexed by binary strings
$z\in\{0,1\}^{\ell}$.  The SELECT implementation computes one-hot equality
flags, applies the corresponding controlled Paulis, and then uncomputes the
flags.  More explicitly, define
\[
  f_z(y):=\mathbf 1[y=z],\qquad
  f(y):=(f_z(y))_{z\in\{0,1\}^{\ell}} .
\]
The decoder $D$ is the reversible map
\begin{equation}\label{eq:select-decoder}
  D:\quad
  \ket y_{\rm addr}\ket{0\cdots0}_{F}
  \longmapsto
  \ket y_{\rm addr}\ket{f(y)}_{F}.
\end{equation}
It can be implemented by a binary tree of AND gates: at level $j$ one computes
whether the first $j$ address bits agree with a prefix.  Each child flag is
obtained from its parent flag and one address literal by a Toffoli gate, with
negated literals supplied by Clifford $X$ gates.  The tree has
$O(2^\ell)$ nodes, hence $D$ uses $O(2^\ell)$ Toffoli gates.

On the flagged space, apply the Clifford unitary
\begin{equation}\label{eq:select-controlled-layer}
  C:\quad
  \ket y_{\rm addr}\ket{f(y)}_{F}\ket\psi
  \longmapsto
  \ket y_{\rm addr}\ket{f(y)}_{F}\,
  \biggl(\prod_{e\in E}(\eta_eP_e)^{f_e(y)}\biggr)\ket\psi .
\end{equation}
For each edge $e$, the flag $f_e$ controls the two single-qubit Paulis in
$P_e$; controlled-$X$, controlled-$Y$, and controlled-$Z$ are Clifford gates.
If $\eta_e=-1$, a $Z$ gate on the flag adds the branch phase.  Therefore
$C$ contributes no $T$ gates.  Since exactly one equality flag is $1$, and since
unused addresses have no flagged Pauli action, we obtain
\[
  (D^\dagger C D)
  \ket y_{\rm addr}\ket{0\cdots0}_{F}\ket\psi
  =
  \ket y_{\rm addr}\ket{0\cdots0}_{F}
  \begin{cases}
    \eta_yP_y\ket\psi, & y\in E,\\
    \ket\psi, & y\notin E,
  \end{cases}
\]
which is exactly \eqref{eq:kitaev-select} with the temporary flags returned to
$\ket{0\cdots0}$.  Finally $2^\ell<2M$, and $D^\dagger$ has the same cost as
$D$.  Since an exact Toffoli has constant Clifford$+T$ cost,
\begin{equation}\label{eq:select-cost}
  T(\mathrm{SEL})
  \le T(D)+T(C)+T(D^\dagger)
  =O(2^\ell)+0+O(2^\ell)
  =O(M)=O(n).
\end{equation}
This is the standard QROM/SELECT construction; more space-efficient tradeoffs
exist \cite{low2024trading} but are not needed here.

\paragraph{Cost and accuracy of PREPARE.}
An arbitrary $\ell$-qubit state can be prepared to Euclidean error $\delta'$
using
\begin{equation}\label{eq:gkw-cost}
  O\!\left(\sqrt{2^{\ell}\log(1/\delta')}+\log(1/\delta')\right)
\end{equation}
$T$ gates together with ancillas \cite{gosset2026stateprep}.  Let
$\ket{\widetilde G}$ be the prepared state and
$\delta'=\|\ket{\widetilde G}-\ket G\|$.  Writing
$B(\ket\phi):=(\bra\phi\otimes I)\mathrm{SEL}(\ket\phi\otimes I)$ and using
$\|\mathrm{SEL}\|=1$,
\[
  \begin{aligned}
  \bigl\|B(\ket{\widetilde G})-B(\ket G)\bigr\|
  &=
  \Bigl\|((\bra{\widetilde G}-\bra G)\otimes I)\,
  \mathrm{SEL}\,(\ket{\widetilde G}\otimes I)
  +(\bra G\otimes I)\,
  \mathrm{SEL}\,((\ket{\widetilde G}-\ket G)\otimes I)\Bigr\| \\
  &\le
  \bigl\|(\bra{\widetilde G}-\bra G)\otimes I\bigr\|
  \cdot\bigl\|\ket{\widetilde G}\otimes I\bigr\|
  +\bigl\|\bra G\otimes I\bigr\|\cdot
  \bigl\|(\ket{\widetilde G}-\ket G)\otimes I\bigr\| \\
  &\le 2\delta' .
  \end{aligned}
\]
Multiplying by $\lambda$, the encoded operator differs from $H_K(J)$ by at most
$2\lambda\delta'$, so choosing
\[
  \delta'=\frac{\eps}{2\lambda}
\]
gives an $\eps$ block encoding.  Since $\lambda\le M=O(n)$ and
$2^{\ell}<2M=O(n)$, \eqref{eq:gkw-cost} becomes
\begin{equation}\label{eq:prepare-cost-final}
  T(\mathrm{PREP})=O\!\left(\sqrt{n\log(n/\eps)}+\log(n/\eps)\right),
\end{equation}
and the inverse preparation costs the same.

The decoder and the state-preparation construction use
The decoder and the state-preparation construction use
\[
O\!\left(
n+\sqrt{n\log(n/\eps)}+\log(n/\eps)
\right)
=
O\!\left(n+\log(1/\eps)\right)
\]
work ancillas in total. work ancillas in total.
Combining Eqs.~\eqref{eq:select-cost} and
\eqref{eq:prepare-cost-final}, and using
$\sqrt{xy}\le\tfrac12(x+y)$ together with $\log n=O(n)$, gives
Eq.~\eqref{eq:kitaev-upper}.

\paragraph{Cost of a controlled block encoding.}
Let
$
U_K=(\mathrm{PREP}^{\dagger}\otimes I)\,
\mathrm{SEL}\,
(\mathrm{PREP}\otimes I).
$
A controlled implementation of $U_K$ can be written as
$
(\mathrm{PREP}^{\dagger}\otimes I)\,
\mathrm{ctrl}\text{-}\mathrm{SEL}\,
(\mathrm{PREP}\otimes I),
$
since the PREPARE operations cancel when the control qubit is zero.
Using the one-hot construction above,
$\mathrm{ctrl}\text{-}\mathrm{SEL}$ is obtained by adding one control
to each flag-controlled Pauli operation. Each such operation has
constant Clifford$+T$ cost, so
\[
T(\mathrm{ctrl}\text{-}\mathrm{SEL})=O(n).
\]
Consequently,
\[
C_{\rm K}
=
O\!\left(
n+\sqrt{n\log(n/\eps_{\mathrm{BE}})}
+\log(n/\eps_{\mathrm{BE}})
\right)
=
O(B_{\rm K}).
\]

\section{Resource Estimates for QSVT}
\label{app:qsp-cost}

\subsection{Block encodings used in the QSVT application}
\label{app:qsp-input-block-encodings}

\subsubsection{Common normalization for the second-quantized family}
\label{app:secondq-common-normalization}

For a fixed Hamiltonian, the construction of Liu \emph{et al.}
depends on the number $L$ of different nonzero interaction coefficients
\cite{liu2025block}. This number may vary within $\cH_n$.
The largest possible table size in this family is
\[
  L_{\max}
  :=
  n^2+\binom{n}{2}^2
  =
  \Theta(n^4).
\]
We use a table with $L_{\max}$ entries for every Hamiltonian and set
all unused entries to zero. These entries do not contribute to the
encoded operator. Since all coefficients have magnitude at most one,
the coefficient scale can also be fixed throughout the family.

Let $\overline\alpha_{2,n}$ be the normalization of this padded
construction. Since the normalization of the full-space construction
is linear in the table size,
\[
  \overline\alpha_{2,n}
  =
  \Theta(L_{\max})
  =
  \Theta(n^4).
\]
The padded table may increase the cost of a sparse instance, but it
does not change the worst-case bound for $\cH_n$. For fixed
$\eps_{\mathrm{BE}}$ and sufficiently large $n$, with
$m_b=O(\log(L_{\max}/\eps_{\mathrm{BE}}))$, the $T$-count remains
\[
  O\!\left(
    n+\sqrt{L_{\max}m_b}
  \right)
  =
  O\!\left(
    n^2\sqrt{
      \log\!\left(\frac{n^4}{\eps_{\mathrm{BE}}}\right)
    }
  \right).
\]

\paragraph{Controlled implementation.}
Ref.~\cite{liu2025block} gives a Clifford gate count
$L_{\max}m_b+o(L_{\max})$ for this construction.
Controlling the elementary Clifford$+T$ circuit gate by gate incurs
only constant $T$-count overhead per gate. Hence
\[
  C_2
  =
  O(L_{\max}m_b)
  =
  O\!\left(
    n^4\log\!\frac{n^4}{\eps_{\mathrm{BE}}}
  \right)
  =
  O(B_2^{2}).
\]

\subsubsection{Common normalization for the Kitaev family}
\label{app:common-normalization}

The original LCU construction has normalization
$\lambda(J):=\sum_{e\in E}|J_e|$. Since $|J_e|\le1$ and
$|E|=M$, we have $\lambda(J)\le M$.

\begin{lemma}
Every $H_K(J)\in\cK_n$ has an LCU block encoding with normalization
$M$. This modification does not change the asymptotic $T$-count or
ancilla count.
\end{lemma}

\begin{proof}
Define
\[
  \mu(J):=\frac{M-\lambda(J)}{2}\ge0.
\]
Adding two cancelling terms gives
\[
  H_K(J)
  =
  \sum_{e\in E}J_eP_e
  +\mu(J)I+\mu(J)(-I).
\]
The Hamiltonian is unchanged, while the total LCU weight is
$\lambda(J)+2\mu(J)=M$. The resulting block encoding therefore has
normalization $M$.

Only two address labels are added, and their SELECT operations are
$I$ and $-I$. Replacing $M$ terms by $M+2$ terms does not change the
asymptotic cost of PREP or SELECT. The $T$-count remains
$O(n+\log(1/\eps_{\mathrm{BE}}))$, and the construction still uses
$O(n+\log(1/\eps_{\mathrm{BE}}))$ ancillas.
\end{proof}

\subsubsection{Hermitian approximation}
\label{app:qsp-hermitian-approximation}
Let $K$ be the operator obtained after approximating the coefficients,
and suppose $\|H-K\|\le\eps_{\mathrm{BE}}$. Define
\[
  \widetilde H:=\frac{K+K^\dagger}{2}.
\]
Then $\widetilde H$ is Hermitian. Moreover, since $H=H^\dagger$,
\[
  \|H-\widetilde H\|
  \le
  \frac{
    \|H-K\|+\|(H-K)^\dagger\|
  }{2}
  \le
  \eps_{\mathrm{BE}}.
\]
The coefficient tables in both constructions can represent
$\widetilde H$ with the same asymptotic cost.

\subsection{Cost of the QSVT circuit}
\label{app:qsp-phase-cost}

Recall that
\[
  \Pi_a:=|0^a\rangle\langle0^a|\otimes I.
\]
The QSVT phase operator
\[
  \Pi_\phi:=e^{i\phi(2\Pi_a-I)}
\]
can be implemented using two $\Pi_a$-controlled NOT gates and one
single-qubit $Z$ rotation
\cite[Fig.~3]{martyn2021grand}. A QSVT sequence with
$N_{\rm calls}$ block-encoding queries contains
$O(N_{\rm calls})$ such operations
\cite[Lemma~19]{gilyen2019qsvt}.

A $\Pi_a$-controlled NOT is an $a$-controlled $X$ gate with negative
controls. With clean work ancillas, standard decompositions implement
it using $O(a)$ Toffoli gates
\cite{barenco1995elementary}. Since an exact Toffoli has constant
$T$-count, the controlled-projector part of the QSVT circuit costs
$O(N_{\rm calls}a)$ $T$ gates.

\paragraph{Oblivious amplitude amplification.}
\label{app:qsp-oaa-cost}
The OAA step includes three uses of $V$ and $V^\dagger$,
together with six projector-controlled NOT gates and three
single-qubit gates
\cite[Theorem~28]{gilyen2019qsvt}.
Both the input and output projectors are
\[
  \Pi_{\mathrm{s}}
  =|0^{a+2}\rangle\langle0^{a+2}|\otimes I.
\]
With clean work ancillas, each projector-controlled NOT
requires $O(a+1)$ $T$ gates.
The three single-qubit gates have the form $e^{-i\phi_j Z}$,
with phases $(-\pi,\pi/2,\pi/2)$
\cite[Lemma~9]{gilyen2019qsvt}.
They are Clifford gates and require no $T$ gates.
Thus the additional OAA cost is $O(a+1)$.
The costs of the calls to $V$ and $V^\dagger$ are already
included in the block-encoding query and phase-synthesis counts.

Combining the cosine and sine sequences requires selecting
the phase angles according to an auxiliary qubit.
Each such rotation can be implemented using two single-qubit
$Z$ rotations and two CNOT gates.
Thus the complete circuit still contains
$O(N_{\rm calls})$ single-qubit $Z$ rotations.
Synthesizing each rotation to error
$O(\eps_{\mathrm{syn}}/N_{\rm calls})$ keeps their total error
within $\eps_{\mathrm{syn}}$ and costs
\[
  O\!\left(
    N_{\rm calls}
    \log\!\frac{N_{\rm calls}}{\eps_{\mathrm{syn}}}
  \right)
\]
$T$ gates \cite{ross2016optimal}.

For both Hamiltonian families, the block-encoding constructions use
$a=O(B_F)$ ancillas. The preceding expression is therefore
\[
  O\!\left(
    N_{\rm calls}
    \left[
      B_F+
      \log\!\frac{N_{\rm calls}}{\eps_{\mathrm{syn}}}
    \right]
  \right).
\]
This proves the upper bound in
Theorem~\ref{thm:qsp-modular-cost}. Its lower bound follows from the
modular cost model, in which every block-encoding query is compiled
separately.

\subsection{Simulation error}
\label{app:qsp-simulation-error}

Let $H$ and $\widetilde H$ be Hermitian and suppose
$\|H-\widetilde H\|\le\eps_{\mathrm{BE}}$. The standard Duhamel bound
gives
\[
  \|e^{-itH}-e^{-it\widetilde H}\|
  \le
  |t|\,\|H-\widetilde H\|
  \le
  |t|\eps_{\mathrm{BE}}.
\]
The QSVT polynomial approximation contributes
$\eps_{\mathrm{QSVT}}$, while phase synthesis contributes
$\eps_{\mathrm{syn}}$. Hence the total simulation error is at most
\[
  |t|\eps_{\mathrm{BE}}
  +
  \eps_{\mathrm{QSVT}}
  +
  \eps_{\mathrm{syn}}.
\]

\end{document}